\documentclass[journal]{IEEEtran}
\usepackage{amsmath,graphicx,amssymb,amsfonts,amsthm}
\usepackage{cite}
\usepackage{algorithm}
\usepackage{algorithmic}
\usepackage[font=scriptsize]{caption}
\usepackage[font=scriptsize]{subcaption}
\usepackage{balance}
\usepackage{mathtools}
\usepackage{stfloats}
\usepackage[usenames, dvipsnames]{color}
\graphicspath{{./Figures/}}

\label{English Chars}

\newcommand{\bp}{{\bf p}}

\theoremstyle{plain}
\newtheorem{theorem}{Theorem}
\newtheorem{proposition}{Proposition}

\newtheorem{corollary}{Corollary}

\theoremstyle{definition}
\newtheorem{definition}{Definition}

\theoremstyle{remark}
\newtheorem{remark}{Remark}

\begin{document}

\title{Cooperative Target Localization in RIS-Enabled ISAC Systems}


\author{Rouhollah Amiri, Abdollah Ajorloo, Mohammad Mahdi Mojahedian, Ahmad Reza Hassanshahi, Fabiola Colone
\thanks{R. Amiri, M. M. Mojahedian and A. R. Hassanshahi are with Department of Electrical Engineering, Sharif University of Technology, Tehran, Iran. A. Ajorloo and F. Colone are with Sapienza University of Rome. (e-mail: amiri@sharif.edu, abdollah.ajorloo@uniroma1.it, mojahedian@sharif.edu, ahmadreza.hassanshahi@ee.sharif.edu, fabiola.colone@uniroma1.it).}
}

\maketitle
\begin{abstract}
This paper develops a cooperative integrated sensing and communication
(ISAC) framework in which a multi-antenna base station (BS)
simultaneously localizes a target and serves multiple communication
users, a subset of which is equipped with reconfigurable intelligent
surfaces (RISs). The Fisher information matrix (FIM) for target
positioning is derived, explicitly characterizing its dependence on the
BS beamforming coefficients, RIS phase profiles, and bistatic sensing
geometry. Under the stated scaling assumptions, coherent RIS phase
alignment yields a Fisher-information gain that scales quadratically
with the number of RIS elements, whereas independent random phases
provide a linear gain in expectation. We formulate a sensing-centric
joint active and passive beamforming problem that minimizes the
position error bound (PEB) subject to per-user
signal-to-interference-plus-noise ratio (SINR) and transmit-power
constraints. The resulting non-convex problem is addressed through an
iterative successive convex approximation (SCA) procedure that solves
a sequence of convex subproblems. We further develop a target
localization estimator that fuses one direct time-of-arrival (ToA)
measurement, one angle-of-arrival (AoA) measurement, and multiple
RIS-assisted indirect ToA measurements. Under small measurement errors
and asymptotically efficient first-stage ToA/AoA estimation, the
estimator covariance approaches the Cram\'er--Rao bound (CRB) to first
order. Numerical results validate the analytical scaling laws and
demonstrate the localization gains enabled by cooperative RIS-equipped
users.
\end{abstract}


\IEEEpeerreviewmaketitle

\vspace{-.5cm}
\section{Introduction} \label{sec:Introduction}
\IEEEPARstart{T}{he} ongoing advancement toward 6G wireless networks is characterized by a migration to higher frequency bands, including millimeter-wave (mmWave) and Terahertz (THz) spectra, enabling radio signals to support both conventional data communication and high-resolution environmental perception \cite{zhang2022terahertz,liu2022integrated}. This dual functionality underpins Integrated Sensing and Communication (ISAC), which seeks to efficiently share spectral and hardware resources while exploiting sensing--communication synergies \cite{liu2022integrated,gonzalez2024integrated,luo2025isac}. ISAC is expected to support applications including autonomous systems, smart cities, and industrial automation \cite{liu2020joint}. Orthogonal Frequency-Division Multiplexing (OFDM), owing to its ubiquity in wireless standards, robustness to multipath fading, and demonstrated efficacy in radar applications, is a compelling waveform for high-precision sensing in ISAC systems \cite{lu2025mimo,su2025isac,li2019multi,keskin2024fundamental}.

Higher frequencies, while enabling high-accuracy sensing and high-rate communication, suffer from severe path loss and susceptibility to blockage. Reconfigurable intelligent surfaces (RISs) can mitigate these limitations by passively reshaping propagation environments via programmable phase shifts and creating virtual LoS paths \cite{wu2021intelligent}. RIS-enhanced ISAC systems can thus improve both communication reliability and sensing accuracy in non-LoS scenarios \cite{meng2024intelligent}, motivating a rapidly growing body of literature on RIS-assisted ISAC.

The existing literature on RIS-assisted ISAC can be broadly
categorized according to the design objective.
Communication-centric studies optimize the BS beamforming and RIS
configuration to improve communication performance while imposing
sensing requirements \cite{luo2022joint,hua2022joint,zhu2024intelligent}.
Sensing-centric works instead optimize metrics such as radar SINR,
the Cram\'{e}r--Rao bound (CRB), or target illumination gain subject
to communication quality-of-service constraints
\cite{shao2022target,song2023intelligent,xu2024joint}.
A broader class of joint designs optimizes the active and passive
beamformers to balance sensing and communication performance through
weighted or constrained formulations
\cite{liu2022proximal,sankar2023beamforming,hua2023secure,
liu2023snr,liu2022joint,wei2023multi,xing2023joint,
meng2023subspace,liu2024fractional}. These studies cover a range of
single- and multi-user/target settings and employ metrics including
radar SINR, mutual information, CRB, and communication rate.

RIS-assisted ISAC has also been extended along several complementary
directions. Energy-efficiency-oriented designs jointly optimize the
BS beamforming and RIS coefficients under sensing requirements
\cite{liu2026energy}; multi-functional RIS architectures incorporate
reflection, refraction, and amplification capabilities
\cite{zhou2026multifunctional}; hybrid beamforming has been developed
for RIS-enabled mmWave OFDM ISAC systems
\cite{wang2025hybrid}; and multi-RIS deployment, including RIS
positions, orientations, and sizes, has been optimized for realistic
mmWave environments \cite{li2026deployment}. Despite these different
objectives and implementations, these works predominantly employ RISs
as separately deployed propagation infrastructure whose configuration
is jointly designed with the BS transmission.

A related but distinct line of work exploits UEs themselves as
cooperative sensing anchors. In \cite{guo2025ue}, UEs either passively
receive BS--target--UE echoes or actively probe the environment and
forward the resulting range information to the BS. In contrast, we
consider RIS-equipped downlink terminals as passive bistatic
illumination nodes: the terminals simultaneously receive communication
service while their RISs create additional
BS--RIS--target--BS sensing paths, with the BS remaining the sensing
receiver. We therefore focus on the localization-information structure
of this topology, its scaling with RIS size and the number of
cooperative terminals, and sensing-centric beamforming under per-user
communication constraints. Specifically, we consider a 2D scenario where a BS equipped with $N_t$ transmit and $N_r$ receive antennas simultaneously localizes a target and serves $\widetilde{K}$ single-antenna users, of which $K < \widetilde{K}$ host RISs with $N_s$ reflecting elements at known locations. These RIS-equipped users receive communication data while creating additional bistatic sensing paths, whereas the remaining $\widetilde{K}-K$ users operate purely in communication mode. We jointly optimize the BS beamforming coefficients and RIS phase shifts to minimize the position error bound (PEB), derived from the CRB, under per-user SINR constraints. The key contributions of this work are as follows:

\begin{itemize}
    \item We consider an ISAC scenario where a subset of users is equipped with RISs to enhance sensing performance, while the remaining users operate purely in communication mode. This framework investigates the trade-off between prioritizing RIS-assisted users, which improve target localization, and communication-centric users, which may degrade sensing quality.

    \item We analyze how RIS phase shifts and beamforming coefficients can be designed to minimize the target position estimation error while ensuring all users meet predefined SINR thresholds. The resulting non-convex problem is handled by an iterative SCA procedure in which convex subproblems are solved successively.

\item We analytically show that the Fisher information contributed by
an RIS-assisted path scales as $N_s^2$ under coherent phase alignment
and as $N_s$ in expectation under independent random phases. We further
derive the localization gain for nonidentical, geometry-dependent
RIS-assisted propagation gains. Provided that the average
RIS-assisted information contribution remains finite and bounded away
from zero as the number of cooperative users increases, the PEB scales
as $1/(N_s\sqrt{K})$ under coherent phase alignment and as
$1/\sqrt{KN_s}$ under independent random phases.


    \item We propose an estimator that localizes the target by combining one noisy angle-of-arrival (AoA) meas\\
    urement, $K$ noisy time-of-arrival (ToA) measurements obtained indirectly via the RIS-assisted users, and one noisy ToA measurement acquired directly from the target. Under small measurement errors and asymptotically efficient first-stage ToA/AoA estimates, the covariance matrix of the proposed estimator approaches the CRB to first order.
\end{itemize}

The rest of the paper is organized as follows. Section~\ref{sec:System model} presents the system model. Section~\ref{sec:Estimation-Theoretic Fundamental Limits} derives the FIM and analyzes the impact of optimal and random RIS phase assignments on the PEB. Section~\ref{sec:Dual-Functional Joint Active and Passive Beamforming} formulates the joint beamforming problem and its local SCA solution. Section~\ref{sec:Efficient Estimator Design} develops the position estimator, and Section~\ref{sec:Numerical Simulations} presents numerical results. Section~\ref{sec:Conclusion} concludes the paper\footnote{Bold lowercase and uppercase letters denote vectors and matrices, respectively. We use $\mathbb{E}\{\cdot\}$, $(\cdot)^\mathsf{T}$, $(\cdot)^*$, $(\cdot)^\mathsf{H}$, $\mathrm{Tr}(\cdot)$, $\mathrm{vec}(\cdot)$, $\|\cdot\|$, and $\|\cdot\|_F$ for expectation, transpose, conjugate, conjugate transpose, trace, vectorization, Euclidean norm, and Frobenius norm, respectively. The operators $\mathrm{diag}(\cdot)$ and $\otimes$ denote diagonalization and the Kronecker product, respectively. We denote the circularly symmetric complex Gaussian and uniform distributions by $\mathcal{CN}(\mu,\sigma^2)$ and $\mathcal{U}(a,b)$, respectively. For any positive integer $L$, let $[L]\triangleq\{1,\ldots,L\}$.}.


\section{System Model}\label{sec:System model}
\subsection{System Configuration}
As illustrated in Fig. \ref{FigConfig}, consider a MIMO BS, the position of which is known and denoted by $\mathbf{q}=[x_q,y_q]^\mathsf{T} \in \mathbb{R}^2$, equipped with an ISAC transmitter with $N_t$ antennas and a radar receiver with $N_r$ antennas operating on a single hardware platform. This system is employed to detect and localize a single target at position $\mathbf{p}=[x_p,y_p]^\mathsf{T} \in \mathbb{R}^2$ and to serve $\tilde{K}$ downlink single-antenna users $ \{\mathcal{U}_k\}_{k=1}^{\tilde{K}} $ with positions $\mathbf{u}_k=[x_{u_k},y_{u_k}]^\mathsf{T} \in \mathbb{R}^2, k=1,\ldots,\tilde{K}$. 
We consider a scenario in which some users $ \{\mathcal{U}_k\}_{k=1}^{K} $ ($ K<\tilde{K} $) are RIS-enabled cooperative terminals (e.g., roadside units or connected vehicles) equipped with $N_s$ reflecting elements and employed to enhance target positioning\footnote{
Note that the classic use-case of RIS in an ISAC system is to serve both the users and target via a separately deployed framework by providing strong NLOS paths. In this scenario, however, we aim to exploit the already-equipped RIS users to enable or enhance the target positioning performance.}.
The localization configuration is depicted in Fig. \ref{FigConfig}.
The position and orientation of RIS-enabled users are considered as a prior knowledge at this stage.

\begin{figure*}[!t]
    \centering
    \includegraphics[width=0.66\textwidth]{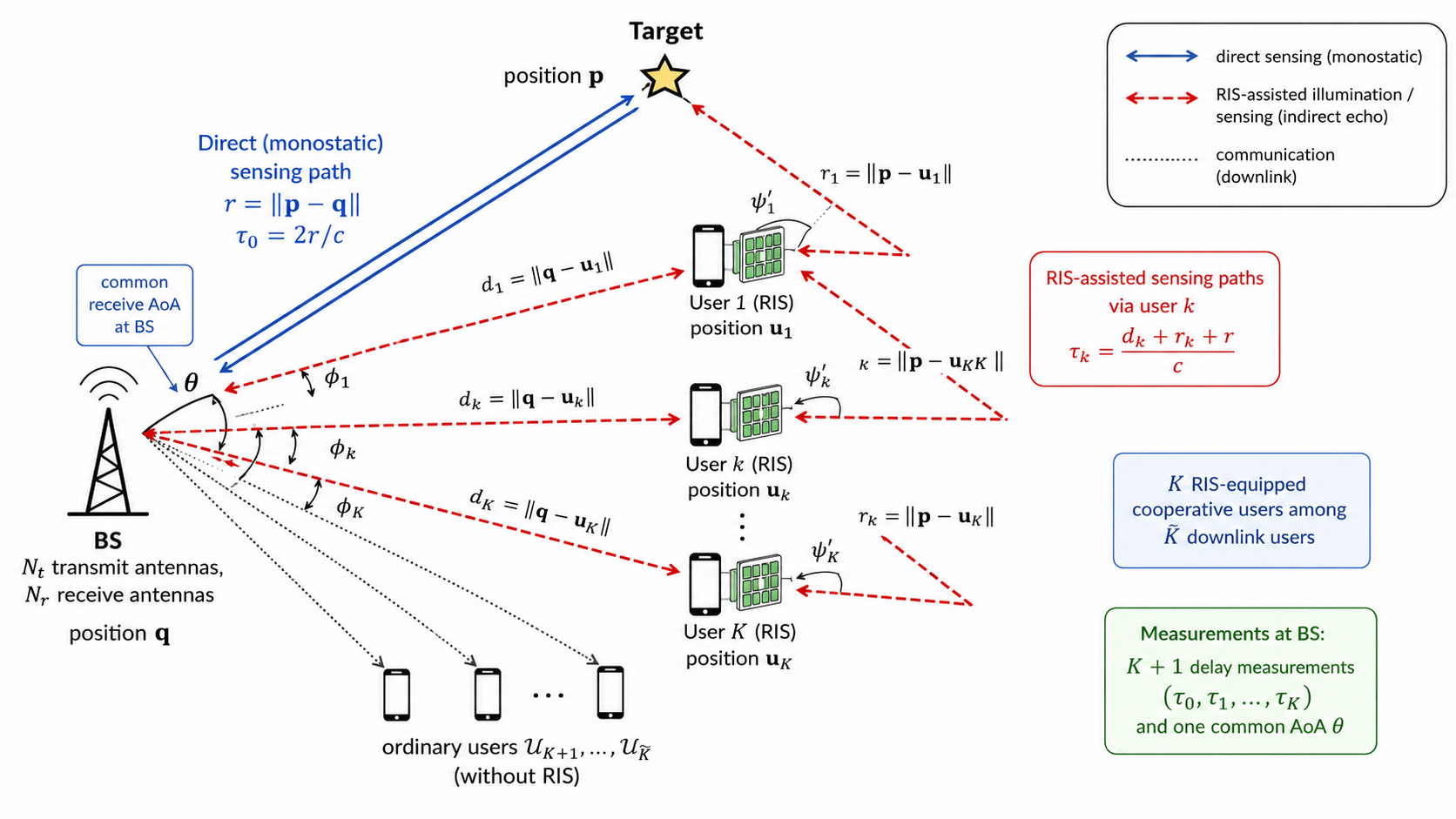}
    \caption{Detailed system configuration. The BS serves $\tilde K$ downlink users, of which $K$ are RIS-equipped cooperative terminals. The direct sensing path is monostatic with delay $\tau_0=2r/c$ and range $r=\|\mathbf{p}-\mathbf{q}\|$. The $k$th RIS-assisted cooperative path follows BS--RIS$_k$--target--BS with delay $\tau_k=(d_k+r_k+r)/c$, BS--RIS distance $d_k=\|\mathbf{q}-\mathbf{u}_k\|$, and target--RIS distance $r_k=\|\mathbf{p}-\mathbf{u}_k\|$. Since every modeled echo makes its final hop from the target to the BS, all resolved paths share the common receive AoA $\theta$ at the BS, whereas the RIS-assisted paths differ in their delays and transmit-side geometry. Dotted links indicate downlink communication only.}
    \label{FigConfig}
\end{figure*}

\subsection{Transmitter Model}\label{subsec:Transmitter Model}
Let $\mathbf{X}[n] \in \mathbb{C}^{N_t\times M}$ be the transmitted ISAC signal matrix for the $n$th OFDM subcarrier, where $-\frac{N}{2} \le n \le \frac{N}{2}$ and $M>N_t$ is the length of signal frame.
The matrix $\textbf{X}[n]$ is given by $\mathbf{X}[n]=\mathbf{W}\mathbf{S}[n]$, where $\mathbf{W}$ denotes the beamforming matrix to be designed to satisfy both radar and communications requirements. The rows of $\mathbf{S}[n] \in \mathbb{C}^{\tilde{K}\times M}$ are modeled as mutually uncorrelated unit-power user streams, i.e., $\mathbb{E}\{\frac{1}{M}\mathbf{S}[n]\mathbf{S}^\mathsf{H}[n]\}=\mathbf{I}_{\tilde{K}}$. Accordingly, the frame-averaged transmit covariance used in the FIM analysis is $\mathbf{R_X}=\mathbb{E}\{\frac{1}{M}\mathbf{X}[n]\mathbf{X}^\mathsf{H}[n]\}=\mathbf{W}\mathbf{W}^\mathsf{H}$. This is a second-order covariance assumption and does not require arbitrary payload symbols to be strictly orthogonal in every realization. The total transmitted power in each subcarrier is given by
\begin{equation}
P_\mathrm{T} \triangleq \mathbb{E}\!\left\{
\frac{1}{M}\mathrm{Tr}\!\left(\mathbf X[n]\mathbf X^\mathsf{H}[n]\right)
\right\}=\mathrm{Tr}(\mathbf W\mathbf W^\mathsf{H})=\|\mathbf W\|_F^2 .
\end{equation}

\subsection{Radar Channel}
As illustrated in Fig.~\ref{FigConfig}, the BS receives the target echo
through one direct (monostatic) path, indexed by $k=0$, and $K$
RIS-assisted (bistatic) paths, indexed by $k=1,\ldots,K$. The direct
path follows BS--target--BS, whereas the $k$th indirect path follows
BS--RIS$_k$--target--BS. Thus, all paths share the final target-to-BS
propagation segment and have the same receive angle of arrival $\theta$;
the RIS-assisted paths have different illumination directions and delays.

Let $r=\|\mathbf p-\mathbf q\|$, $d_k=\|\mathbf q-\mathbf u_k\|$,
and $r_k=\|\mathbf p-\mathbf u_k\|$ denote the BS--target,
BS--RIS$_k$, and RIS$_k$--target distances, respectively. 
propagation speed. 
Further, let $\varphi_k$ denote the azimuth of RIS$_k$ as
viewed from the BS; $\theta$ is the azimuth of the target as viewed
from the BS.

Each RIS controls the phase of its $N_s$ reflecting elements through
the diagonal matrix $\mathbf\Omega_k=\mathrm{diag}(\boldsymbol\omega_k)$,
where
\begin{align}\label{eq:Omega_k}
	\boldsymbol\omega_k
	= [e^{j\omega_{k,0}},\ldots,e^{j\omega_{k,N_s-1}}]^\mathsf{T}.
\end{align}
The phase profile $\boldsymbol\omega_k$ is
a design variable which is supposed to be optimized (e.g., based on an estimated target position from a
previous snapshot processing).
Then, let $\varphi_k'$ and $\psi_k'$ be, respectively, the incident angle of
the BS signal and the departure angle toward the target, measured in
the local array coordinates of RIS$_k$. 
With $\mathbf b(\alpha)$
denoting the RIS steering vector in direction $\alpha$, the coherent
response along the BS--RIS$_k$--target path is
\begin{equation}
	g_k=\mathbf b^\mathsf{H}(\psi_k')\mathbf\Omega_k
	\mathbf b(\varphi_k').
\end{equation}  

We model the BS transmit array, each RIS, and the BS receive array as
uniform linear arrays with half-wavelength inter-element spacing,
centered at their respective reference points. Their steering
vectors are denoted by $\mathbf a(\alpha)\in\mathbb C^{N_t}$,
$\mathbf b(\alpha)\in\mathbb C^{N_s}$, and
$\mathbf c(\alpha)\in\mathbb C^{N_r}$, respectively. For example, the
transmit steering vector is \cite{van2002optimum}
\begin{align}\label{eq: steering vector}
	\mathbf a(\alpha)
	= \left[e^{-j\pi\frac{N_t-1}{2}\sin\alpha},\ldots,
	e^{j\pi\frac{N_t-1}{2}\sin\alpha}\right]^\mathsf{T},
\end{align}
and $\mathbf b(\alpha)$ and $\mathbf c(\alpha)$ have the same form
with $N_t$ replaced by $N_s$ and $N_r$, respectively.

The path coefficients account for the propagation loss and, for the
indirect paths, the coherent RIS response. Assuming isotropic array
elements with unit antenna gains, they are given by
\begin{align}
	\gamma_0
	&= \left(\frac{\lambda^2}{(4\pi)^3r^4L_0}\right)^{1/2},
	\\
	\gamma_k
	&= g_k\left(\frac{\zeta\lambda^2}
	{(4\pi)^4d_k^2r_k^2r^2L_k}\right)^{1/2},
	\quad k=1,\ldots,K,
\end{align}
where $\lambda$ is the wavelength, $\zeta$ is the bistatic radar
cross section of the RIS elements, $L_0$ accounts for additional
two-way attenuation and system losses on the direct path, and $L_k$
is the corresponding loss factor for the $k$th indirect path.

Modeling the target as a point scatterer \cite{skolnik1962introduction},
we can now express the target response matrices
$\mathbf G_k\in\mathbb C^{N_r\times N_t}$ as
\begin{align}\label{eq: Channel G_k}
	\mathbf G_0
	&= \gamma_0\mathbf c(\theta)\mathbf a^\mathsf{H}(\theta),
	\nonumber\\
	\mathbf G_k
	&= \gamma_k\mathbf c(\theta)\mathbf a^\mathsf{H}(\varphi_k),
	\quad k=1,\ldots,K.
\end{align}
The common receive steering vector reflects the shared target-to-BS
segment, whereas the transmit steering vector identifies the BS
illumination direction for each path.\\

\subsection{Radar Receiver}
The received signal for the $n$th subcarrier is given by
\begin{align}\label{eq:Received Signal Radar}
	\mathbf{Y}_{\mathrm{R}}[n] = \sum_{k=0}^{K} \beta_k e^{-j2\pi n\Delta f \tau_k} \mathbf{G}_k \mathbf{X}[n] + \mathbf{Z}_{\mathrm{R}}[n]
\end{align}
where $\beta_k \in \mathbb{C}$ accounts for the target reflection coefficient and any other scaling factor in the $ k $th path not included in the corresponding target response matrix. 
The term $\Delta f=B/(N+1)$ denotes the subcarrier spacing for a total bandwidth of $B$, and $\tau_k$ stands for the time of arrival (TOA) for the $k$th path.
The matrix $\mathbf{Z}_{\mathrm{R}}[n]\in \mathbb{C}^{N_r\times M}$ represents additive white complex Gaussian noise. Throughout the FIM derivation, $\sigma_R^2$ denotes the variance of each real/imaginary component; equivalently, the complex-sample covariance is $2\sigma_R^2\mathbf I_{N_r}$. This convention is used consistently in Appendix~\ref{sec:Calculation of the FIM}.

\subsection{Communications Receiver}
The received signal at the communication side, for subcarrier $n$, can be modeled as
\begin{align}\label{eq: Received Signal Comm}
\mathbf{Y}_{\mathrm{C}}[n] = \mathbf{HX}[n] + \mathbf{Z}_{\mathrm{C}}[n]
\end{align}
where $\mathbf{Z}_{\mathrm{C}}[n]\in \mathbb{C}^{\tilde{K} \times M}$ represents a white complex Gaussian noise with covariance matrix $\sigma^2_{\mathrm{C}} \mathbf{I}_{\tilde{K}}$, and $\mathbf{H}=[\mathbf{h}_1,\mathbf{h}_2,\ldots,\mathbf{h}_{\tilde{K}}]^\mathsf{H} \in \mathbb{C}^{{\tilde{K}}\times N_t}$ denotes the communication channel matrix with independently distributed entries, which is considered to be known to the BS.

\subsection{Performance Metrics}
We consider the CRB as a metric for target parameter estimation accuracy, which is a bound on the covariance matrix of any unbiased estimator.
Based on the received signal model in \eqref{eq:Received Signal Radar}, the parameter vector $\boldsymbol{\gamma}\in \mathbb{R}^{2K+4}$ includes the target position and the nuisance reflection coefficients\footnote{We aim to locate the target based on TOA and AOA measurements, i.e., $\tau_k$'s and $\theta$, and relation of $\beta_k$'s on $\mathbf{p}$ will not be exploited; thus, $\beta_k$'s will be considered as nuisance parameters.}, i.e.,
\begin{align}\label{eq:gamma}
\boldsymbol{\gamma} = [\mathbf{p}^\mathsf{T},\boldsymbol{\beta}^\mathsf{T}]^\mathsf{T}
\end{align}
where $\boldsymbol{\beta}=[\boldsymbol{\beta}_0^\mathsf{T},\boldsymbol{\beta}_1^\mathsf{T},\ldots,\boldsymbol{\beta}_K^\mathsf{T}]^\mathsf{T}$, each of which containing the real and imaginary parts as $\boldsymbol{\beta}_k = [\boldsymbol{\beta}_k^{\mathcal{R}}, \boldsymbol{\beta}_k^{\mathcal{I}}]^\mathsf{T}$.

Let $\hat{\boldsymbol{\gamma}}$ denote an estimate of $\boldsymbol{\gamma}$ based on observation $\mathbf{Y}_{\mathrm{R}}[n]$'s.
The mean squared error (MSE) matrix of $\hat{\boldsymbol{\gamma}}$ satisfies the following inequality \cite{kay}
\begin{align}\label{eq: Ineq 1}
\mathbb{E}\{(\hat{\boldsymbol{\gamma}}-\boldsymbol{\gamma})(\hat{\boldsymbol{\gamma}}-\boldsymbol{\gamma})^\mathsf{T} \}\succeq\mathbf{J}_{\boldsymbol{\gamma}}^{-1}
\end{align}
where $\mathbf{J}_{\boldsymbol{\gamma}}$ is the Fisher information matrix (FIM) of the parameter vector $\boldsymbol{\gamma}$.
Let $\hat{\mathbf{p}}$ be an estimate of the target position, \eqref{eq: Ineq 1} gives a lower bound on the corresponding MSE as
\begin{align}\label{eq: Ineq 2}
\mathbb{E}\{(\hat{\mathbf{p}}-\mathbf{p})(\hat{\mathbf{p}}-\mathbf{p})^\mathsf{T} \}\succeq \text{CRB}(\mathbf{p})
\end{align}
where $\text{CRB}(\mathbf{p}) =  \left[\mathbf{J}_{\boldsymbol{\gamma}}^{-1}\right]_{1:2, 1:2}$.
Different functions of $\text{CRB}(\mathbf{p})$ can be used as the performance metric, and the typical choices are trace, maximum eigenvalue and determinant, which also referred to as A-optimality, E-optimality and D-optimality metrics, respectively.
In this paper, we consider A-optimality as an estimation-theoretic metric for ISAC joint active and passive beamforming design.

Furthermore, we consider the per-user SINR as a communication-related metric.
By representing the beamforming matrix as $\mathbf{W}=[\mathbf{w}_1,\mathbf{w}_2,\ldots,\mathbf{w}_{\tilde{K}}]$, the associated SINR for the $k$th user can be written as
\begin{align}\label{eq: SINR_k}
\text{SINR}_k = \frac{|\mathbf{h}_k^\mathsf{H}\mathbf{w}_k|^2}{\sum\nolimits_{i=1,i\neq k}^{\tilde{K}} |\mathbf{h}_k^\mathsf{H}\mathbf{w}_i|^2 + \sigma^2_C}
\end{align}

As the focus of the proposed scheme is on sensing-centric ISAC, the
RIS-enabled terminals are utilized primarily to enhance sensing
accuracy rather than communication performance. In particular, the
RIS phase profiles are designed to coherently steer the reflected BS
signal toward the sensing target, rather than toward the communication
receiver. Therefore, the RIS-mediated field is not intentionally
beamformed to enhance the BS--user communication link. Accordingly, in
the derivation of \eqref{eq: SINR_k}, RIS-mediated BS--user paths and
local RIS-to-receiver coupling are neglected. This is an explicit
architecture assumption: the externally mounted RIS and the
communication antenna are assumed to be sufficiently isolated, and the
conventional BS--user channel is assumed to dominate the communication
link.\footnote{The analysis further assumes (i) perfect BS knowledge of
	$\mathbf{H}$ and the RIS positions/orientations; (ii) synchronization
	and calibration sufficient for coherent processing; (iii) dominant
	resolvable specular sensing paths with diffuse clutter and unmodeled
	multipath neglected; (iv) known Swerling-I RCS statistics, but not
	known instantaneous reflection coefficients; and (v) ideal continuous
	RIS phases and dominant LoS BS--RIS/RIS--target links. Relaxing these
	assumptions is left for future work.}


\section{Estimation-Theoretic Fundamental Limits}\label{sec:Estimation-Theoretic Fundamental Limits}
In this section, we aim to analyze the fundamental limits of target localization by deriving the corresponding FIM.

\subsection{Equivalent FIM of Target Position}
To facilitate the analysis, we consider a mapping from $\boldsymbol{\gamma}$ to an intermediate parameter vector, i.e., the target position-related channel parameters denoted by $\boldsymbol{\eta}=[\boldsymbol{\tau}^\mathsf{T},\theta,\boldsymbol{\beta}^\mathsf{T}]^\mathsf{T} \in \mathbb{R}^{3K+4}$, where $\boldsymbol{\tau}=[\tau_0,\tau_1,\ldots,\tau_K]^\mathsf{T}$ and $\boldsymbol{\beta}$ is defined below \eqref{eq:gamma}.
The channel-parameter vector $\boldsymbol{\eta}$ is a differentiable function of $\boldsymbol{\gamma}$. Using the FIM chain rule, we obtain
\begin{align}\label{eq:J_eta}
\mathbf{J}_{\boldsymbol{\gamma}}=\mathbf{T}\mathbf{J}_{\boldsymbol{\eta}}\mathbf{T}^\mathsf{T}
\end{align}
where $\mathbf{J}_{\boldsymbol{\eta}}$ stands for the FIM of $\boldsymbol{\eta}$, and $\mathbf{T}$ is the associated Jacobian matrix, which are given, respectively, by
\begin{align}\label{eq:J_eta partiotioned}
\mathbf{J}_{\boldsymbol{\eta}} \triangleq 
\left[ \begin{matrix}
{{\mathbf{\Psi }}_{1}} & {{\mathbf{\Psi }}_{2}}  \\
\mathbf{\Psi }_{2}^{T} & {{\mathbf{\Psi }}_{3}}  \\
\end{matrix} \right]
\end{align}
and
\begin{align} \label{eq:Jacobian}
\mathbf{T} \triangleq \left(\frac{\partial \boldsymbol{\eta}}{\partial \boldsymbol{\gamma}^\mathsf{T}}\right)^\mathsf{T} = 
\left[ \begin{matrix}
\mathbf{\Xi } & \mathbf{O}  \\
\mathbf{O} & {{\mathbf{I}}_{2(K+1)}}  \\
\end{matrix} \right]
\end{align}
where the block matrices ${{\mathbf{\Psi }}_{1}}$, ${{\mathbf{\Psi }}_{2}}$ and ${{\mathbf{\Psi }}_{3}}$ are derived in Appendix \ref{sec:Calculation of the FIM}, and $\mathbf{\Xi}$ is given by
\begin{align}\label{eq:Xi}
\mathbf{\Xi} = \left[2c^{-1}\mathbf{e}_0,c^{-1}\mathbf{e}_1,\ldots,c^{-1}\mathbf{e}_K,r^{-1}\boldsymbol{\rho}\right]
\end{align}
where $\mathbf{e}_k=[\cos\theta+s_k \cos\varphi_k, \sin\theta+s_k \sin\varphi_k]^\mathsf{T}$ (the corresponding angles are shown in Fig. \ref{FigConfig}) with $s_k$ being a binary variable which is zero for $k=0$ and is one otherwise, $\boldsymbol{\rho} = [-\sin\theta,\cos\theta]^\mathsf{T}$, $r=\|\mathbf{p}-\mathbf{q}\|$ and $c$ denotes the wave propagation speed.
\begin{proposition} 
The CRB of target position is given by
\begin{align}\label{eq:CRB(p)}
\mathrm{CRB}(\mathbf{p}) = (\mathbf{\Xi}\mathbf{\Psi}_1\mathbf{\Xi}^\mathsf{T}-\mathbf{\Xi}\mathbf{\Psi}_2\mathbf{\Psi}_3^{-1}\mathbf{\Psi}_2^\mathsf{T}\mathbf{\Xi}^\mathsf{T})^{-1}
\end{align}
\end{proposition}

\begin{IEEEproof} 
After direct matrix multiplications of \eqref{eq:J_eta}, the CRB of $\boldsymbol{\gamma}$ is obtained as
\begin{align}\label{eq:CRB of gamma}
\text{CRB}(\boldsymbol{\gamma}) = \left[ \begin{matrix}
\mathbf{\Xi}\mathbf{\Psi}_1\mathbf{\Xi}^\mathsf{T} & \mathbf{\Xi}\mathbf{\Psi}_2  \\
\mathbf{\Psi}_2^\mathsf{T}\mathbf{\Xi}^\mathsf{T} & {{\mathbf{\Psi }}_{3}}  \\
\end{matrix} \right]^{-1}
\end{align}
and $\text{CRB}(\mathbf{p})$ is the $2\times 2$ upper-diagonal block of $\mathbf{J}_{\boldsymbol{\gamma}}^{-1}$. Invoking block-matrix inversion formula on \eqref{eq:CRB of gamma} \cite{horn2012matrix} immediately completes the proof. \end{IEEEproof}

\begin{definition} 
The inner term in \eqref{eq:CRB(p)} is a kind of FIM degraded due to presence of nuisance parameter $\boldsymbol{\beta}$, which is referred to as equivalent FIM (EFIM), and denoted by $\mathbf{J}_e(\mathbf{p})$.
\end{definition}

\begin{definition} 
	The general ranging information (GRI) is defined as a $2\times 2$ matrix of the form $\lambda \mathbf{J}_\text{r}(\theta,\varphi,s)$, where $\lambda$ is a non-negative parameter denoting the so-called information intensity, and $\mathbf{J}_\text{r}(\theta,\varphi,s)$ is the associated direction matrix, given by
	\begin{align}\label{eq:General Ranging Information}
		\mathbf{J}_\text{r}(\theta,\varphi,s) \triangleq
		\left[ \begin{matrix}
			x^2 & xy \\
			xy & y^2  \\
		\end{matrix} \right]
	\end{align}
	where $x=(2-s)\cos\theta+s\cos\varphi$, $y=(2-s)\sin\theta+s\sin\varphi$ and $s$ is a binary variable.
	Note that this definition covers the notion of ranging direction for monostatic ranging ($s=0$) and bistatic ranging ($s=1$). $\mathbf{J}_\text{r}$ has only one non-zero eigenvalue with associated eigenvector $\mathbf{e}=[x,y]^\mathsf{T}$.
\end{definition}

\begin{definition}  The angular information (AI) is defined as a $2\times 2$ matrix of the form $\delta \mathbf{J}_\text{a}(\theta)$, where $\delta$ is a non-negative parameter denoting the so-called information intensity, and $\mathbf{J}_\text{a}(\theta)$ is the corresponding direction matrix, given by
\begin{align}\label{eq:Angular Information}
\mathbf{J}_\text{a}(\theta) \triangleq
\left[ \begin{matrix}
\sin^2\theta & -\sin\theta \cos\theta\\
-\sin\theta \cos\theta & \cos^2\theta  \\
\end{matrix} \right]
\end{align}
which has only one unit eigenvalue with associated eigenvector $\boldsymbol{\rho}=[-\sin\theta,\cos\theta]^\mathsf{T}$.
\end{definition}

\begin{theorem} 
When the $K+1$ paths are sufficiently separated in delay so that their cross-delay correlation terms are negligible (with $|\tau_i-\tau_k|\gtrsim 1/B$ serving as a practical resolution rule), the EFIM of the target position is well approximated by
\begin{align} \label{eq:EFIM Approx Final Form}
	\mathbf{J}_\text{e}(\mathbf{p}) \approx \sum_{k=0}^{K} \lambda_k \mathbf{J}_\text{r}(\theta,\varphi_k,s_k) + \delta_k \mathbf{J}_\text{a}(\theta)
\end{align}
where $\lambda_k$ and $\delta_k$ are the corresponding information intensities for the $k$th path, given by
\begin{align}
\lambda_k & = \frac{MN(N\!+\!2)}{(N\!+\!1)}\frac{\pi^2 B^2 |\beta_k|^2}{3c^2 \sigma_R^2} \ell_k \label{eq:lambda_k}\\
\delta_k & = M(N\!+\!1)\frac{|\beta_k|^2}{r^2 \sigma_R^2} (m_k-\frac{|n_k|^2}{\ell_k}) \label{eq:delta_k}
\end{align}
where
\begin{align}
\ell_k & \triangleq \mathrm{Tr}(\mathbf{G}_k\mathbf{R_X}\mathbf{G}_k^\mathsf{H}) \label{eq:a_i def}\\
m_k & \triangleq \mathrm{Tr}(\mathbf{\dot{G}}_k\mathbf{R_X}\mathbf{\dot{G}}_k^\mathsf{H}) \label{eq:b_i def}\\
n_k & \triangleq \mathrm{Tr}(\mathbf{G}_k\mathbf{R_X}\mathbf{\dot{G}}_k^\mathsf{H}) \label{eq:c_i def}
\end{align}
and $\mathbf{R_X}=\mathbf{WW}^\mathsf{H}$.
\end{theorem}

\begin{IEEEproof} 
The proof is given in Appendix \ref{sec:Proof of Theorem 1}. Furthermore, a simplified version of $\ell_k, m_k$ and $n_k$ are presented, in terms of corresponding steering vectors and their derivatives as well as the active and passive beamforming weights, at the end of this appendix.
\end{IEEEproof}

\begin{definition}
 Using the calculated EFIM, the positioning error bound (PEB), which determines the best achievable accuracy of any unbiased estimator, is defined as
\begin{align}\label{eq:PEB}
\mathcal{P} \triangleq \sqrt{\mathrm{Tr}(\mathbf{J}_\text{e}^{-1}(\mathbf{p}))}
\end{align}
\end{definition}

\begin{remark}
The intensity $\lambda_k$ depends on the signal bandwidth ($B$), SNR of the path (including the RIS gain embedded in $\ell_k$ for $k\ne0$), the integration time ($M$), the number of subcarriers ($N+1$), and also the relative geometry of BS, RIS-enabled users and the target. It also depends on the lengths of the transmit array ($N_t$), receive array ($N_r$), and RIS ($N_s$), as made more explicit in \eqref{eq:C0}.
The intensity $\delta_k$ depends on similar parameters in which signal bandwidth is replaced with inverse target-BS range.
\end{remark}
 
\begin{remark} 
The information direction matrices in \eqref{eq:EFIM Approx Final Form} are in three main categories:
1) The monostatic range direction (MRD) matrix, i.e., $\mathbf{J}_\text{r}$ for $k=0$, has only one eigenvector along with the BS to the target direction.
2) The angular direction matrix $\mathbf{J}_\text{a}$ has only one eigenvector orthogonal to that of the MRD.
3) Each bistatic range direction matrix, i.e., $\mathbf{J}_\text{r}$ for $k\ne0$, has only one eigenvector along with the combination of BS-target and user-target directions, or equivalently the direction orthogonal to the ellipse with the BS and the corresponding user as its focal points.
The linear combination of these information direction matrices with the presented intensities as \eqref{eq:EFIM Approx Final Form} yields the corresponding EFIM.
For a centered half-wavelength ULA and an isotropic transmit covariance, the angular-to-ranging information ratio for an indirect path is approximately
\begin{align}
\frac{\delta_k}{\lambda_k} = \frac{c^2(N+1)^2}{4B^2r^2N(N+2)}(N_r^2-1)\cos^2\theta, \quad k\ge 1.
\end{align}
Hence, at medium-to-long ranges and sufficiently large occupied bandwidth, the ranging term can dominate. To obtain transparent scaling laws and a tractable beamforming design, we use the following ranging-dominant approximation in Sections III-B and IV:
\begin{align} \label{eq:EFIM Approx Final Form (TOA only)}
	\mathbf{J}_\text{e}(\mathbf{p}) \approx \sum_{k=0}^{K} \lambda_k \mathbf{J}_\text{r}(\theta,\varphi_k,s_k).
\end{align}
The full delay-and-angle measurement model is retained in the localization estimator of Section~\ref{sec:Efficient Estimator Design}.
\end{remark}

\begin{remark}
As inferred from \eqref{eq:EFIM Approx Final Form}, the presence of multiple paths, while they can be resolved, enhances the corresponding EFIM.
When the paths are not resolvable, however, several other terms will also be included in \eqref{eq:EFIM Approx Final Form}, which, due to their semi-random phases leads to a severe fading-like effect, and highly degrades the overall EFIM. Thus, this paradigm is most useful when the occupied bandwidth and geometry make the direct and RIS-assisted delays sufficiently resolvable; carrier frequency alone does not determine delay resolution.
\end{remark}

\begin{remark}
Under the Swerling-I model for the RCS fluctuations \cite{swerling1960probability,skolnik1962introduction}, $|\beta_k|^2$ follows an exponential distribution with parameter $\mathcal{E}$. Taking the expectation of the ranging-dominant EFIM in \eqref{eq:EFIM Approx Final Form (TOA only)} gives
\begin{align} \label{eq:Expected EFIM}
	\bar{\mathbf{J}}_\text{e}(\mathbf{p}) = \mathbb{E}_{\boldsymbol{\beta}}\left\{\mathbf{J}_\text{e}(\mathbf{p})\right\} = \sum_{k=0}^{K} \bar{\lambda}_k \mathbf{J}_\text{r}(\theta,\varphi_k,s_k)
\end{align}
where $\bar{\lambda}_k$ follows \eqref{eq:lambda_k} with $|\beta_k|^2$ replaced by $\mathcal{E}$.
\end{remark}

\subsection{User-Aided Target Localization Accuracy Gain}

We now characterize the average localization accuracy enabled by the
RIS-assisted paths and investigate its scaling with respect to the
number of RIS-enabled users $K$ and the number of reflecting elements
$N_s$. The target and users are assumed to be uniformly placed in the covered area such that $\theta,\varphi_k\sim \mathcal{U}(0,2\pi)$. To eliminate any preferred spatial direction in the averaged
localization performance, we assume that the target and RIS-enabled
users follow a jointly rotationally invariant spatial distribution
around the BS. Furthermore, the target RCS fluctuations are modeled
according to the Swerling-I model.

We first derive the averaged EFIM while allowing the direct and
RIS-assisted paths to experience different geometry-dependent
propagation gains. Based on this result, we then characterize the
corresponding PEB gain and identify the conditions under which the
$K$- and $N_s$-scaling laws are preserved.



\begin{proposition}
\label{prop:avg_EFIM}
(Average EFIM)
	
	For the purpose of characterizing the $K$- and $N_s$-scaling, assume the rotationally invariant geometry model above, Swerling-I RCS fluctuations, and an isotropic average transmit covariance $\mathbb{E}\{\mathbf R_X\}=\frac{P_\mathrm{T}}{N_t}\mathbf I_{N_t}$. For the $k$th RIS-assisted path, define
	\begin{align}\label{eq:gamma_tilde}
		\widetilde{\gamma}_k \triangleq \frac{\gamma_k}{g_k}, \qquad k=1,\ldots,K,
	\end{align}
	such that
	\begin{align}
		|\gamma_k|^2=|\widetilde{\gamma}_k|^2|g_k|^2, \qquad k=1,\ldots,K.
	\end{align}
	Then, the EFIM averaged over the propagation geometry, RCS fluctuations, and, when applicable, random RIS phases is
	\begin{align}\label{eq:Asymptotic EFIM}
		\bar{\mathbf{J}}_\text{e}
		=
		C_0
		\left(
		2\bar{\Gamma}_0
		+
		\xi\sum_{k=1}^{K}\bar{\Gamma}_k
		\right)
		\mathbf{I}_2,
	\end{align}
	where
	\begin{align}\label{eq:C0}
		C_0
		=
		\frac{\pi^2MN(N\!+\!2)N_rB^2}
		{3(N\!+\!1)c^2}
		\frac{\mathcal{E}P_\mathrm{T}}{\sigma_R^2},
	\end{align}
	and
	\begin{align}
		\bar{\Gamma}_0
		&\triangleq
		\mathbb{E}\left\{|\gamma_0|^2\right\},
		\label{eq:Gamma0_general}
		\\
		\bar{\Gamma}_k
		&\triangleq
		\mathbb{E}\left\{
		|\widetilde{\gamma}_k|^2
		\left[1+\cos(\theta-\varphi_k)\right]
		\right\},
		\quad k=1,\ldots,K.
		\label{eq:Gammak_general}
	\end{align}
	Furthermore, $\xi=N_s^2$ when the optimal phase profile is adopted, i.e.,
	\begin{align}\label{eq:Optimal Phase RIS Asymptotic Analysis}
		\boldsymbol{\omega}_{k,i}
		=
		\frac{\pi}{2}(2i
		\!-\!N_s\!+\!1)
		(\sin\psi'_k-\sin\varphi'_k),
		\quad i=0,\ldots,N_s\!-\!1,
	\end{align}
	and $\xi=N_s$ when independent random RIS phases uniformly distributed over $[0,2\pi]$ are considered.

\end{proposition}

\begin{IEEEproof}
	Under $\mathbb{E}\{\mathbf R_X\}=\frac{P_\mathrm{T}}{N_t}\mathbf I_{N_t}$, \eqref{eq:Simplified l_k,m_k,n_k} gives, conditioned on the propagation geometry,
	\begin{align}\label{eq:Cal Simplied}
		\mathbb{E}\{\ell_0\}
		&=
		P_\mathrm{T}N_r|\gamma_0|^2,
		\nonumber\\
		\mathbb{E}\{\ell_k\}
		&=
		P_\mathrm{T}N_r|\widetilde{\gamma}_k|^2|g_k|^2,
		\quad k=1,\ldots,K.
	\end{align}
	
	Substituting \eqref{eq:Cal Simplied} into \eqref{eq:lambda_k}, averaging over the Swerling-I RCS fluctuations, and using the ranging-dominant EFIM in \eqref{eq:EFIM Approx Final Form (TOA only)} yield
	\begin{align}\label{eq:Je bar}
		\bar{\mathbf{J}}_\text{e}
		=
		C_0
		\bigg[
		&
		\mathbb{E}
		\left\{
		|\gamma_0|^2
		\mathbf{J}_\text{r}(\theta,\varphi_0,0)
		\right\}
		\nonumber\\
		&+
		\sum_{k=1}^{K}
		\mathbb{E}
		\left\{
		|\widetilde{\gamma}_k|^2
		|g_k|^2
		\mathbf{J}_\text{r}(\theta,\varphi_k,1)
		\right\}
		\bigg],
	\end{align}
	where $C_0$ is given in \eqref{eq:C0}. Note that $\varphi_0$ can be chosen arbitrarily since $\mathbf{J}_\text{r}(\theta,\varphi_0,0)$ is independent of $\varphi_0$.
	
	Under the rotationally invariant geometry assumption, the direct-path contribution satisfies
	\begin{align}
		\mathbb{E}
		\left\{
		|\gamma_0|^2
		\mathbf{J}_\text{r}(\theta,\varphi_0,0)
		\right\}
		=
		2\bar{\Gamma}_0\mathbf{I}_2.
	\end{align}
	
	For the $k$th RIS-assisted path, the corresponding ranging direction vector is
	\begin{align}
		\mathbf{e}_k
		=
		[\cos\theta+\cos\varphi_k,\,
		\sin\theta+\sin\varphi_k]^\mathsf{T},
	\end{align}
	whose squared norm is
	\begin{align}
		\|\mathbf{e}_k\|^2
		=
		2\left[1+\cos(\theta-\varphi_k)\right].
	\end{align}
	Hence, rotational averaging gives
	\begin{align}
		\mathbb{E}
		\left\{
		|\widetilde{\gamma}_k|^2
		\mathbf{J}_\text{r}(\theta,\varphi_k,1)
		\right\}
		=
		\bar{\Gamma}_k\mathbf{I}_2,
	\end{align}
	where $\bar{\Gamma}_k$ is defined in \eqref{eq:Gammak_general}.
	
	We next consider two cases for the RIS phase profile.
	
	\textit{Case 1:} For the optimal phase profile in \eqref{eq:Optimal Phase RIS Asymptotic Analysis}, the phases of all RIS elements are coherently aligned, yielding
	\begin{align}
		|g_k|^2=N_s^2.
	\end{align}
	Hence, $\xi=N_s^2$.
	
	\textit{Case 2:} For independent random RIS phases uniformly distributed over $[0,2\pi]$, the cross terms vanish in expectation, giving
	\begin{align}
		\mathbb{E}_{\boldsymbol{\omega}_k}\{|g_k|^2\}=N_s.
	\end{align}
	Hence, $\xi=N_s$.
	
	Combining the above results gives
	\begin{align}\label{eq:Je bar 2}
		\bar{\mathbf{J}}_\text{e}
		=
		C_0
		\left(
		2\bar{\Gamma}_0
		+
		\xi\sum_{k=1}^{K}\bar{\Gamma}_k
		\right)
		\mathbf{I}_2,
	\end{align}
	which is \eqref{eq:Asymptotic EFIM}.
\end{IEEEproof}

\begin{theorem}
	Under the assumptions of Proposition~\ref{prop:avg_EFIM}, the PEB
	associated with the averaged EFIM satisfies
	\begin{align}\label{eq:PEB Accuracy Gain}
		\frac{\bar{\mathcal{P}}}
		{\bar{\mathcal{P}}_{\text{Unc.}}}
		=
		\sqrt{
			\frac{2\bar{\Gamma}_0}
			{2\bar{\Gamma}_0
				+\xi\sum_{k=1}^{K}\bar{\Gamma}_k}
		},
	\end{align}
	where $\bar{\mathcal{P}}_{\text{Unc.}}$ corresponds to $K=0$.
\end{theorem}

\begin{IEEEproof}
	According to \eqref{eq:Je bar 2}, the averaged EFIM for the case
	without RIS-enabled users is
	\begin{align}
		\bar{\mathbf{J}}_{\text{e,Unc.}}
		=
		2C_0\bar{\Gamma}_0\mathbf{I}_2.
	\end{align}
	Applying the PEB definition in \eqref{eq:PEB} to
	$\bar{\mathbf{J}}_{\text{e,Unc.}}$ and \eqref{eq:Je bar 2} gives
	\eqref{eq:PEB Accuracy Gain}.
\end{IEEEproof}

\begin{corollary}
	\label{cor:PEB_scaling}
	Define the average RIS-assisted information contribution as
	\begin{align}\label{eq:Average RIS Information Contribution}
		\bar{\Gamma}_{\mathrm{RIS},K}
		\triangleq
		\frac{1}{K}\sum_{k=1}^{K}\bar{\Gamma}_k.
	\end{align}
	Then, \eqref{eq:PEB Accuracy Gain} can be equivalently written as
	\begin{align}\label{eq:PEB Accuracy Gain Scaling}
		\frac{\bar{\mathcal{P}}}
		{\bar{\mathcal{P}}_{\text{Unc.}}}
		=
		\sqrt{
			\frac{2\bar{\Gamma}_0}
			{2\bar{\Gamma}_0
				+K\xi\bar{\Gamma}_{\mathrm{RIS},K}}
		}.
	\end{align}
	Suppose that there exist constants
	$0<\Gamma_{\min}\leq\Gamma_{\max}<\infty$, independent of $K$,
	such that
	\begin{align}\label{eq:Average Gain Bounded}
		\Gamma_{\min}
		\leq
		\bar{\Gamma}_{\mathrm{RIS},K}
		\leq
		\Gamma_{\max}
	\end{align}
	for sufficiently large $K$. Then, as $K\xi\rightarrow\infty$,
	\begin{align}\label{eq:PEB Scaling Law}
		\frac{\bar{\mathcal{P}}}
		{\bar{\mathcal{P}}_{\text{Unc.}}}
		=
		\begin{cases}
			\Theta\left(\dfrac{1}{N_s\sqrt{K}}\right),
			& \text{coherent RIS phase alignment},\\[2mm]
			\Theta\left(\dfrac{1}{\sqrt{KN_s}}\right),
			& \text{independent random RIS phases}.
		\end{cases}
	\end{align}
\end{corollary}

\begin{IEEEproof}
	From \eqref{eq:Average RIS Information Contribution},
	\begin{align}
		\sum_{k=1}^{K}\bar{\Gamma}_k
		=
		K\bar{\Gamma}_{\mathrm{RIS},K}.
	\end{align}
	Substituting this relation into \eqref{eq:PEB Accuracy Gain} gives
	\eqref{eq:PEB Accuracy Gain Scaling}. Under
	\eqref{eq:Average Gain Bounded}, the quantity
	$\bar{\Gamma}_{\mathrm{RIS},K}$ remains bounded above and bounded
	away from zero as $K$ increases. Therefore, for
	$K\xi\rightarrow\infty$,
	\begin{align}
		\frac{\bar{\mathcal{P}}}
		{\bar{\mathcal{P}}_{\text{Unc.}}}
		=
		\Theta\left(\frac{1}{\sqrt{K\xi}}\right).
	\end{align}
	Using $\xi=N_s^2$ for coherent RIS phase alignment and
	$\xi=N_s$ for independent random RIS phases gives
	\eqref{eq:PEB Scaling Law}.
\end{IEEEproof}

\begin{remark}
Although \eqref{eq:PEB Scaling Law} characterizes the asymptotic scaling with $K$, the finite-$K$ localization gain is also affected by the relative strength of the direct and RIS-assisted paths. In particular, when the direct path is dominant, as is often the case in practice, the contribution of a small number of RIS-enabled users may remain relatively modest, and the asymptotic $K$-scaling becomes apparent only after the aggregate information provided by the RIS-assisted paths becomes sufficiently significant. Hence, a stronger direct path mainly delays the onset of the asymptotic regime, while it does not change the scaling exponent provided that its average contribution remains finite and independent of $K$.	
	
	The condition in \eqref{eq:Average Gain Bounded} does not require
	identical propagation gains across the RIS-assisted paths. It only
	requires their average information contribution to remain finite and
	non-vanishing as the number of RIS-enabled users increases. Therefore,
	the individual quantities $\bar{\Gamma}_k$ may differ due to different
	BS--RIS and RIS--target distances, loss factors, and bistatic
	geometries. 
\end{remark}
%

\begin{remark}
 Note that the optimal phase adjustment of RIS as \eqref{eq:Optimal Phase RIS Asymptotic Analysis} requires a certain level of BS-Radar cooperation as well as a prior knowledge of target position. As previously mentioned, the corresponding optimal RIS gain is $\xi=N_s^2$.
It may be interesting that when the users are equipped with RIS for another purpose so that the RIS phase profiles are adjusted in an uncooperative manner or even randomly, the proposed scheme is still useful, and provides a gain $\xi=N_s$.
\end{remark}

Another problem to be investigated is the effect of optimal beamforming designed to satisfy both the radar and communication requirements.


\section{Dual-Functional Joint Active and Passive Beamforming}
\label{sec:Dual-Functional Joint Active and Passive Beamforming}

We consider the sensing-centric joint active and passive beamforming
problem
\begin{align}\label{eq:Opt Problem General Form}
&\underset{\mathbf{W},\{\boldsymbol{\omega}_{k}\}_{k=1}^{K}}
{\min}\quad \mathcal{A}(\mathbf{p}) \\
&\text{s.t.}\quad
\mathrm{SINR}_k\geq\varepsilon_k,\quad
k=1,\ldots,\widetilde{K},
\nonumber\\[-1mm]
&\hspace{14mm}\|\mathbf W\|_F^2\leq P_\mathrm{T}, \nonumber
\end{align}
where
\[
\mathcal{A}(\mathbf p)
=
\mathrm{Tr}\!\left(
\mathbb E_{\boldsymbol{\beta}}
\{\mathbf J_\mathrm{e}(\mathbf p)\}^{-1}
\right)
\]
is the A-optimality criterion, $\varepsilon_k$ is the required SINR
of user $k$, and $P_\mathrm{T}$ is the transmit-power budget.

\begin{remark}
The design in \eqref{eq:Opt Problem General Form} requires a nominal target position, which can be supplied in tracking operation by the estimate from the previous coherent processing interval \cite{bar2004estimation}. If unavailable, a worst-case or Bayesian design over the surveillance region can instead be adopted.
\end{remark}

Under the communication model of Section~\ref{sec:System model}, the
SINR and power constraints in \eqref{eq:Opt Problem General Form} are
independent of the RIS phases. Hence, the passive beamforming design
decouples from the active beamforming problem.

\begin{proposition}
\label{prop:RIS_phase}
The optimum phase profile of the $k$th RIS is
\begin{align}
\label{eq:RISPhaseOptimal}
\boldsymbol{\omega}_{k,i}^{*}
=
\frac{\pi}{2}(2i-N_s+1)
(\sin\psi'_k-\sin\varphi'_k),
\quad i=0,\ldots,N_s-1.
\end{align}
\end{proposition}

\begin{IEEEproof}
From \eqref{eq:lambda_k}, \eqref{eq:a_i def}, and the path-gain model,
the information intensity of the $k$th RIS-assisted path is
proportional to $|g_k|^2$. Since
$\mathbf J_\mathrm{r}(\theta,\varphi_k,s_k)\succeq\mathbf0$ and the
communication constraints are independent of the RIS phases, the
objective is minimized by maximizing each $|g_k|^2$. The coherent
phase alignment in \eqref{eq:Optimal Phase RIS Asymptotic Analysis}
then gives \eqref{eq:RISPhaseOptimal} and $|g_k|^2=N_s^2$.
\end{IEEEproof}

We next optimize the active beamforming for the resulting RIS gains.
To express the ranging-only EFIM in a compact form, define
$\mathbf d_i=[x_i,y_i]^\mathsf{T}$ as
\begin{align}
\mathbf d_0
&=2[\cos\theta,\sin\theta]^\mathsf{T}, \nonumber\\
\mathbf d_i
&=[\cos\theta+\cos\varphi_i,\,
   \sin\theta+\sin\varphi_i]^\mathsf{T},
\quad i=1,\ldots,K,
\label{eq:directionCompact}
\end{align}
and
\begin{align}
\mathbf A_0
&=\mathbf a_\theta\mathbf a_\theta^\mathsf{H},\qquad
\mathbf A_i
=\mathbf a(\varphi_i)\mathbf a(\varphi_i)^\mathsf{H},
\quad i=1,\ldots,K, \nonumber\\
\alpha_0
&=\gamma_0^2c_1\|\mathbf c_\theta\|^2,\qquad
\alpha_i
=|\gamma_i|^2c_1\|\mathbf c_\theta\|^2,
\quad i=1,\ldots,K,
\label{eq:alphaAi}
\end{align}
where
$c_1=MN(N+2)\pi^2B^2\mathcal E/ [3(N+1)c^2\sigma_R^2]$.

Indeed, using \eqref{eq:lambda_k} with $|\beta_i|^2$ replaced by
$\mathcal E$, the quantities $\ell_i$ in \eqref{eq:a_i def} can be
written as
\begin{align}
\ell_0
&=
\gamma_0^2\|\mathbf c_\theta\|^2
\mathrm{Tr}(\mathbf A_0\mathbf R_X), \nonumber\\
\ell_i
&=
|\gamma_i|^2\|\mathbf c_\theta\|^2
\mathrm{Tr}(\mathbf A_i\mathbf R_X),
\quad i=1,\ldots,K.
\label{eq:ellCompact}
\end{align}
Furthermore,
$\mathbf J_\mathrm{r}(\theta,\varphi_i,s_i)
=\mathbf d_i\mathbf d_i^\mathsf{T}$.
Hence, under the ranging-dominant approximation of
\eqref{eq:EFIM Approx Final Form (TOA only)}, the expected EFIM in
\eqref{eq:Expected EFIM} can be written as
\begin{equation}
\label{eq:compactEFIM}
\bar{\mathbf J}_\mathrm{e}(\mathbf p)
=
\sum_{i=0}^{K}
\alpha_i
\mathrm{Tr}(\mathbf A_i\mathbf R_X)
\mathbf d_i\mathbf d_i^\mathsf{T},
\qquad
\mathbf R_X=
\sum_{k=1}^{\widetilde K}
\mathbf w_k\mathbf w_k^\mathsf{H}.
\end{equation}
Since $\bar{\mathbf J}_\mathrm{e}$ is $2\times2$, $\mathcal A(\mathbf p)=\mathrm{Tr}(\bar{\mathbf J}_\mathrm{e})/\det(\bar{\mathbf J}_\mathrm{e})$. Let
\begin{align}
\mathbf r&=\mathrm{vec}(\mathbf R_X),\qquad
\mathbf a_i=\mathrm{vec}(\mathbf A_i), \nonumber\\
\mathbf a
&=
\sum_{i=0}^{K}
\alpha_i(x_i^2+y_i^2)\mathbf a_i, \nonumber\\
\widetilde{\mathbf B}
&=
\sum_{i=0}^{K}\sum_{j=0}^{K}
\alpha_i\alpha_j
(x_i y_j-y_i x_j)^2
\mathbf a_i\mathbf a_j^\mathsf{H}.
\label{eq:compactAB}
\end{align}
Defining
$\chi_i=\mathrm{Tr}(\mathbf A_i\mathbf R_X)$, it follows from
\eqref{eq:compactEFIM} that
\begin{align}
\mathrm{Tr}(\bar{\mathbf J}_\mathrm{e})
&=
\sum_{i=0}^{K}
\alpha_i(x_i^2+y_i^2)\chi_i
=
\frac{1}{2}
\left(
\mathbf r^\mathsf{H}\mathbf a+
\mathbf a^\mathsf{H}\mathbf r
\right), \nonumber\\
\det(\bar{\mathbf J}_\mathrm{e})
&=
\frac{1}{2}
\sum_{i=0}^{K}\sum_{j=0}^{K}
\alpha_i\alpha_j
(x_i y_j-y_i x_j)^2
\chi_i\chi_j
=
\frac{1}{2}
\mathbf r^\mathsf{H}
\widetilde{\mathbf B}\mathbf r .
\label{eq:traceDetCompact}
\end{align}
Hence,
\begin{equation}
\label{eqObjAopt4}
\mathcal A(\mathbf p)
=
\frac{\mathbf r^\mathsf{H}\mathbf a+
      \mathbf a^\mathsf{H}\mathbf r}
     {\mathbf r^\mathsf{H}
      \widetilde{\mathbf B}\mathbf r},
\end{equation}
where the denominator is positive whenever the ranging-only EFIM is
nonsingular.

For the active design, define $\mathbf W_k=\mathbf w_k\mathbf w_k^\mathsf{H}$ and $\mathbf H_k=\mathbf h_k\mathbf h_k^\mathsf{H}$. With these lifted variables, the SINR constraints become affine in
$\{\mathbf W_k\}$. The relation
$\mathbf W_k=\mathbf w_k\mathbf w_k^\mathsf{H}$ is equivalent to $\mathbf W_k\succeq\mathbf 0$ and $\mathrm{rank}(\mathbf W_k)=1$. We relax the rank-one constraints and solve the resulting semidefinite relaxation. If the beamforming vectors are retained explicitly, the equality $\mathbf W_k=\mathbf w_k\mathbf w_k^\mathsf{H}$ can equivalently be relaxed as
\begin{equation}
\mathbf W_k\succeq
\mathbf w_k\mathbf w_k^\mathsf{H}
\quad\Longleftrightarrow\quad
\begin{bmatrix}
\mathbf W_k & \mathbf w_k\\
\mathbf w_k^\mathsf{H} & 1
\end{bmatrix}
\succeq\mathbf 0,
\label{eq:SchurRelax}
\end{equation}
by the Schur complement. Rank-one recovery is applied whenever the
resulting $\mathbf W_k$ is not rank one.

Introducing auxiliary variables $t$ and $\widetilde t$, the relaxed
fractional problem can be written in epigraph form as
\begin{subequations}
\label{eqOptPrecoding2}
\begin{align}
\min_{\mathbf r,\{\mathbf W_k\},t,\widetilde t}&\quad
 t \nonumber\\
\mathrm{s.t.}\quad
&\mathbf r=
\mathrm{vec}\!\left(
\sum_{k=1}^{\widetilde K}\mathbf W_k
\right),
\label{eqOptPrecoding2a}\\
&
\mathrm{Tr}(\mathbf H_k\mathbf W_k)
-\varepsilon_k
\!\!\sum_{\substack{i=1\\i\neq k}}^{\widetilde K}
\mathrm{Tr}(\mathbf H_k\mathbf W_i)
\geq
\varepsilon_k\sigma_C^2,
\quad k\in[\widetilde K],
\label{eqOptPrecoding2b}\\
&
\sum_{k=1}^{\widetilde K}
\mathrm{Tr}(\mathbf W_k)
\leq P_\mathrm{T},
\qquad
\mathbf W_k\succeq\mathbf 0,\quad
k\in[\widetilde K],
\label{eqOptPrecoding2c}\\
&
\begin{bmatrix}
t&1\\
1&\widetilde t
\end{bmatrix}
\succeq\mathbf 0,
\label{eqOptPrecoding2d}\\
&
\widetilde f(\widetilde{\mathbf r})\leq0,
\label{eqConstNonCvx1}
\end{align}
\end{subequations}
where $\widetilde{\mathbf r}=[\mathbf r^\mathsf{T},\widetilde t]^\mathsf{T}$ and
\begin{equation}
\label{eqConstNonCvx}
\widetilde f(\widetilde{\mathbf r})
=
\widetilde t
(\mathbf r^\mathsf{H}\mathbf a+
 \mathbf a^\mathsf{H}\mathbf r)
-
\mathbf r^\mathsf{H}
\widetilde{\mathbf B}\mathbf r .
\end{equation}
All constraints in \eqref{eqOptPrecoding2} are convex except
\eqref{eqConstNonCvx1}. We therefore adopt a local SCA procedure
\cite{sun2016majorization}. At iteration $l$, the non-convex function
is approximated around
$\widetilde{\mathbf r}_l=[\mathbf r_l^\mathsf{T},
\widetilde t_l]^\mathsf{T}$ as
\begin{align}
\widehat f_l(\widetilde{\mathbf r})
&=
\widetilde f(\widetilde{\mathbf r}_l)
+
\Re\!\left\{
\nabla\widetilde f(\widetilde{\mathbf r}_l)^\mathsf{H}
(\widetilde{\mathbf r}-\widetilde{\mathbf r}_l)
\right\},
\label{eq:SCAapprox}\\
\nabla\widetilde f(\widetilde{\mathbf r}_l)
&=
\begin{bmatrix}
2\widetilde t_l\mathbf a
-2\widetilde{\mathbf B}\mathbf r_l\\
\mathbf r_l^\mathsf{H}\mathbf a+
\mathbf a^\mathsf{H}\mathbf r_l
\end{bmatrix}.
\label{eq:SCAgrad}
\end{align}
Replacing \eqref{eqConstNonCvx1} by
$\widehat f_l(\widetilde{\mathbf r})\leq0$ yields the convex
subproblem solved at each iteration.

If the resulting $\mathbf W_k$ are rank one, the beamforming vectors
are obtained directly. Otherwise, rank-one candidates are generated
using principal-eigenvector extraction or Gaussian randomization, and
only candidates satisfying the original SINR and power constraints are
retained. The complete procedure is summarized in
Algorithm~\ref{alg:SCA}.

\begin{algorithm}[t]
\caption{Iterative Joint Active and Passive Beamforming}
\label{alg:SCA}
\begin{algorithmic}[1]
\STATE Set $\{\boldsymbol{\omega}_k\}$ according to
Proposition~\ref{prop:RIS_phase} and form the corresponding
coefficients in \eqref{eq:alphaAi}.
\STATE Initialize a feasible
$\widetilde{\mathbf r}_0=
[\mathbf r_0^\mathsf{T},\widetilde t_0]^\mathsf{T}$
and set $l=0$.
\REPEAT
\STATE Solve \eqref{eqOptPrecoding2} with
\eqref{eqConstNonCvx1} replaced by
$\widehat f_l(\widetilde{\mathbf r})\leq0$.
\STATE Obtain $\widetilde{\mathbf r}_{l+1}$ and
$\{\mathbf W_{k,l+1}\}_{k=1}^{\widetilde K}$, and compute
$\mathcal A_{l+1}$ from \eqref{eqObjAopt4}.
\STATE $l\leftarrow l+1$.
\UNTIL{
$|\mathcal A_l-\mathcal A_{l-1}|/
\mathcal A_{l-1}<\delta$
\textbf{ or } $l>L_{\max}$}
\STATE Recover rank-one beamforming candidates and retain only those
satisfying the original SINR and power constraints.
\RETURN
$\mathbf W=[\mathbf w_1,\ldots,\mathbf w_{\widetilde K}]$
and $\{\boldsymbol{\omega}_k\}_{k=1}^{K}$.
\end{algorithmic}
\end{algorithm}

The above procedure is used as a local  method. Since the
first-order approximation in \eqref{eq:SCAapprox} has not been shown
to globally majorize the non-convex constraint
\eqref{eqConstNonCvx1}, we do not claim guaranteed MM/KKT convergence.
Termination is based on stabilization of the original A-optimal
objective, and the original SINR and power constraints are explicitly
rechecked after rank-one recovery. A convergence-guaranteed SCA
variant is left for future work.


\section{Efficient Estimator Design}\label{sec:Efficient Estimator Design}
In this section, we aim to develop an efficient estimator for the so-called target localization problem.

\subsection{Problem Formulation}

Based on signal model in \eqref{eq:Received Signal Radar}, the target position-related parameters are the TOAs $\tau_k$'s and the AOA $\theta$, which can be extracted from the received signal at the radar receiver using local delay and angle estimation algorithms such as \cite{delay_doppler_estimation,AOA_estimation}.
The localization problem then can be defined as the estimation of target position based on the aforementioned measurements.

The total measurement model can be represented as
\begin{align}\label{eq:total meas. model}
\tilde{\mathbf{m}} = \mathbf{m} + \Delta \mathbf{m}
\end{align}
where $\mathbf{m}=[\boldsymbol{\tau}^\mathsf{T},\theta]^\mathsf{T} \in \mathbb{R}^{K+2}$ denoting the true measurement vector and $\boldsymbol{\tau}=[\tau_0,\tau_1,\ldots,\tau_K]^\mathsf{T}$. The vectors $\tilde{\mathbf{m}}$ and $\Delta \mathbf{m}$ are the corresponding observed version and the measurement noise vector which is assumed to be modeled as a zero-mean Gaussian random vector with covariance matrix $\mathbf{C_m}$.

Considering the best possible accuracy of the measurement estimation algorithms determined by the associated CRB, $\mathbf{C_m}$ can be written, by applying block matrix inversion formula on \eqref{eq:J_eta}, as 
\begin{align}\label{eq:C_m}
\mathbf{C_m} = (\mathbf{\Psi}_1 - \mathbf{\Psi}_2\mathbf{\Psi}_3^{-1}\mathbf{\Psi}_2^\mathsf{T})^{-1}
\end{align}

The noise-free measurement parameters in $\mathbf{m}$ can be defined as
\begin{align}
\tau_0 & = \frac{2}{c} \|\mathbf{p}-\mathbf{q}\| \label{eq:tau_0}\\
\tau_k & = \frac{1}{c} (\|\mathbf{q}-\mathbf{u}_k\|+ \|\mathbf{p}-\mathbf{u}_k\| + \|\mathbf{p}-\mathbf{q}\|),\quad k \ne 0 \label{eq:tau_k}\\
\theta & = \operatorname{atan2}(y_p-y_q,x_p-x_q) \label{eq:theta}
\end{align}

Equations \eqref{eq:tau_0}-\eqref{eq:theta} induce, respectively, a circular, $K$ elliptical and a linear locus on which the target is located.
Due to highly nonlinear and nonconvex nature of the associated maximum likelihood (ML) problem, finding its globally optimal solution is a challenging task.
In the rest of this section, we construct a weighted pseudo-linear approximation of this problem; the resulting QCQP is a generalized trust-region subproblem (GTRS) and can be solved globally for the approximated model.

\begin{remark}
When each delay path is resolved in delay domain, an important problem is to determine which delay measurement is associated with the $k$th user (Since the direct-path delay ($\tau_0$) is the smallest one, it corresponds to the BS already). Such a data association problem is commonplace in multipath-aided positioning and multi-sensor multi-target tracking, and efficient and sophisticated methods exist for it (e.g., refer to \cite{bar2004estimation}).
\end{remark}

\begin{remark}
The opportunistic use of RIS-assisted user measurements effectively provides additional bistatic illumination viewpoints at the users' locations, which can enhance positioning performance. Furthermore, in a range-only formulation that deliberately neglects AOA, the additional bistatic measurements can provide the second independent ranging direction needed for 2-D localization.
\end{remark}

\subsection{Proposed Estimator}\label{subsec:Proposed Estimator}
Rewriting \eqref{eq:tau_0} and \eqref{eq:tau_k} as $r_0=\|\mathbf{p}-\mathbf{q}\|$ and $r_k=\|\mathbf{p}-\mathbf{u}_k\|$, respectively, and squaring both sides, yield
\begin{align}
\|\mathbf{q}\|^2-r_0^2-2\mathbf{q}^\mathsf{T}\mathbf{p} + \|\mathbf{p}\|^2 &= 0 \label{eq:tau_0 equation}\\
\|\mathbf{u}_k\|^2-r_k^2-2\mathbf{u}_k^\mathsf{T}\mathbf{p} + \|\mathbf{p}\|^2 &= 0, \,\, k=1,2,\ldots,K \label{eq:tau_k equation}
\end{align}
where $r_0=\frac{c\tau_0}{2}$ and $r_k=c(\tau_k-\frac{\tau_0}{2})-\|\mathbf{q}-\mathbf{u}_k\|$ for $k\ne 0$.

Equation \eqref{eq:theta} can be recast, by taking tangent and cross-multiplication, as
\begin{align} \label{eq:theta equation}
\boldsymbol{\rho}^\mathsf{T}\mathbf{q} - \boldsymbol{\rho}^\mathsf{T}\mathbf{p} = 0
\end{align}
where $\boldsymbol{\rho}=[-\sin \theta , \cos \theta]^\mathsf{T}$.

Stacking the equations in \eqref{eq:tau_0 equation}, \eqref{eq:tau_k equation} and \eqref{eq:theta equation} gives the following matrix equation:
\begin{align}\label{eq:d-A model}
\mathbf{d}-\mathbf{A}\boldsymbol{\vartheta}=0
\end{align}
where $\boldsymbol{\vartheta}=[\mathbf{p}^\mathsf{T},\|\mathbf{p}\|^2]^\mathsf{T}$ and
\begin{align}
\mathbf{d} & = [\|\mathbf{q}\|^2\!-\!r_0^2,\|\mathbf{u}_1\|^2\!-\!r_1^2,\ldots,\|\mathbf{u}_K\|^2\!-\!r_K^2,\boldsymbol{\rho}^\mathsf{T}\mathbf{q}]^\mathsf{T} \nonumber\\
\mathbf{A} & =
\left[ \begin{matrix}
2\mathbf{S}^\mathsf{T} & -\mathbf{1}_{K+1}  \\
\boldsymbol{\rho}^\mathsf{T} & 0  \\
\end{matrix} \right], \quad \mathbf{S} = [\mathbf{q},\mathbf{u}_1,\ldots,\mathbf{u}_K]
\end{align}

By considering the measurement noise, the matrix equation in \eqref{eq:d-A model} can be approximated up to linear noise terms as
\begin{align}\label{eq:d-A noisy model}
\tilde{\mathbf{d}}-\tilde{\mathbf{A}}\boldsymbol{\vartheta}=\mathbf{B}\Delta \mathbf{m}
\end{align}
where $\tilde{\mathbf{d}}$ and $\tilde{\mathbf{A}}$ are noisy versions of $\mathbf{d}$ and $\mathbf{A}$ obtained by replacing the true terms with their measured counterparts, and
\begin{align}
\mathbf{B} = \left[ \begin{matrix}
-c\,r_0 & \mathbf{0}_K^\mathsf{T}  & 0\\
c\,\mathbf{r} & -2c\,\text{diag}(\mathbf{r})  & \mathbf{0}_K\\
0 & \mathbf{0}_K^\mathsf{T} & r_0
\end{matrix} \right], \mathbf{r} = [r_1,r_2,\ldots,r_K]^\mathsf{T}
\end{align}

Based on the first-order pseudo-linear model in \eqref{eq:d-A noisy model}, we use the following weighted least-squares approximation:
\begin{align}\label{eq:CWLS initial}
& \underset{\boldsymbol{\vartheta}}{\mathop{\min }}\,{{\left( \mathbf{\tilde{d}-\tilde{A}}\boldsymbol{\vartheta} \right)}^{T}}\mathbf{\Gamma }\left( \mathbf{\tilde{d}-\tilde{A}}\boldsymbol{\vartheta} \right) \nonumber\\ 
& \text{s}\text{.t}\text{. }\boldsymbol{\vartheta}(3)={{\boldsymbol{\vartheta}}^{2}}(1)+{{\boldsymbol{\vartheta}}^{2}}(2)
\end{align}
where $\mathbf{\Gamma}=(\mathbf{B}\mathbf{C}_{\mathbf{m}}\mathbf{B}^\mathsf{T})^{-1}$. Since $\mathbf B$ contains the unknown ranges, it is formed in practice from the measured/coarse range estimates; this plug-in weighting converges to the ideal weighting in the small-noise regime. The problem can then be recast as
\begin{align}\label{eq:CWLS GTRS}
& \underset{\boldsymbol{\vartheta}}{\mathop{\min }}\,   \boldsymbol{\vartheta}^\mathsf{T}\mathbf{\tilde{D}}\boldsymbol{\vartheta} - 2\mathbf{\tilde{f}}^\mathsf{T}\boldsymbol{\vartheta} \nonumber\\
& \text{s}\text{.t}\text{. }
\boldsymbol{\vartheta}^\mathsf{T}\mathbf{E}\boldsymbol{\vartheta} + 2 \mathbf{g}^\mathsf{T} \boldsymbol{\vartheta} = 0
\end{align}
where $\mathbf{\tilde{D}} = \mathbf{\tilde{A}}^\mathsf{T}\mathbf{\Gamma}\mathbf{\tilde{A}}$, $\mathbf{\tilde{f}}=\mathbf{\tilde{A}}^\mathsf{T}\mathbf{\Gamma}\mathbf{\tilde{d}}$, $\mathbf{E}=\text{diag}([\mathbf{1}_2^\mathsf{T},0])$ and $\mathbf{g}=[\mathbf{0}_2^\mathsf{T},-0.5]^\mathsf{T}$.

The optimization problem in \eqref{eq:CWLS GTRS} is a QCQP with one quadratic constraint and belongs to the class of problems referred to as generalized trust region subproblems (GTRS) \cite{mor1993generalizations}, which possesses necessary and sufficient optimality conditions, and thus its globally optimal solution can be found through efficient algorithms.
        
\begin{proposition}
The global minimizer of \eqref{eq:CWLS GTRS},     $\boldsymbol{\vartheta}^{\ast}$ can be found as
\begin{equation} \label{eq: vartheta optimal}
{\boldsymbol\vartheta ^*}(\lambda) = {\left( {{\bf{\tilde{D}}} + \lambda {\bf{E}}} \right)^{ - 1}}\left( {{\bf{\tilde{f}}} - \lambda \mathbf{g}} \right)
\end{equation}
such that
\begin{equation} \label{eq: PSD}
{\bf{\tilde{D}}} + \lambda {\bf{E}} \succeq 0
\end{equation}
where $\lambda\in\mathbb{R}$ is the root of the following constraint equation in the interval satisfying \eqref{eq: PSD}:
\begin{equation} \label{eq: P(lambda)}
\mathcal{P}(\lambda)=\boldsymbol{\vartheta}^*(\lambda)^\mathsf{T}\mathbf{E}\boldsymbol{\vartheta}^*(\lambda) + 2 \mathbf{g}^\mathsf{T} \boldsymbol{\vartheta}^*(\lambda).
\end{equation}
\end{proposition}

\begin{IEEEproof} The proof follows the Karush Kuhn Tucker (KKT) conditions for the GTRS problem in \eqref{eq:CWLS GTRS}. See \cite[Theorem 3.2]{mor1993generalizations} for details.
\end{IEEEproof}

Note that only one of the candidate $\lambda$'s is associated with the global solution ${\boldsymbol\vartheta^*}$, and finding this specific $\lambda$ via root finding of $\mathcal{P}(\lambda)$ is not an easy task since it involves taking the inverse of a matrix, parametrically. However, findings of \cite[Theorem 5.2]{mor1993generalizations} suggest that the task of obtaining the optimal $\lambda$ can be accomplished easily as explained ahead.

Condition \eqref{eq: PSD} specifies an interval for which ${\bf{\tilde{D}}} + \lambda {\bf{E}}$ is positive semidefinite. It can be easily verified that this condition is equivalent to
\begin{equation} \label{eq: Equival. condit.}
\lambda  \ge \tilde \lambda  =  - 1/\max ({\rm{eig}}({\bf{E}},{\bf{\tilde{D}}}))
\end{equation} 
where ${\rm{eig}}({\bf{E}},{\bf{\tilde{D}}})$ is the set of generalized eigenvalues of the matrix pair $({\bf{E}},{\bf{\tilde{D}}})$ \cite{beck2008exact}. It is shown by \cite[Theorem 5.2]{mor1993generalizations} that for $\lambda  \ge \tilde \lambda $, $\mathcal{P}(\lambda )$ is a strictly decreasing function, which immediately implies that $\mathcal{P}(\lambda )$ has exactly one root in the interval $[\tilde \lambda ,\infty )$ that can be easily found using MATLAB fzero routine, or even a simpler method such as the bisection-algorithm.

After finding the optimal $\lambda$, the corresponding ${\boldsymbol\vartheta ^*}$ can be obtained from \eqref{eq: vartheta optimal}, from which the final position estimate is obtained as $\hat{\mathbf{p}}= {\boldsymbol\vartheta ^*}(1:2)$.

\begin{remark}
The users' positions may be subject to uncertainties due to previous estimation errors. Thus, considering the statistics of such an uncertainty, one can simply reformulate the localization problem which will be in the same form as \eqref{eq:CWLS initial} by replacing $\mathbf{\Gamma}$ with another weighting matrix accounting for both the measurement noise and user location uncertainty distributions.
Further analysis of this issue is omitted for brevity.
\end{remark}

\subsection{Performance Analysis}\label{subsec:Performance Analysis}
In the following, we aim to establish efficiency of the proposed estimator under mild noise conditions by comparing the estimated target covariance matrix with the CRB.

The estimator in \eqref{eq:CWLS GTRS} is a constrained weighted least-squares problem obtained from a first-order pseudo-linearization. We therefore characterize its covariance in the small-noise regime using the constrained CRB associated with the linearized model, rather than claiming exact maximum-likelihood equivalence for the original nonlinear measurements.

\begin{definition}
$\mathbf{U}$ is defined as a matrix whose columns form an orthonormal basis for the null space of the full row rank matrix ${\bf{F}}(\boldsymbol\vartheta)=\boldsymbol{\vartheta}^\mathsf{T}\mathbf{E} + \mathbf{g}^\mathsf{T}$. In other words,
\begin{equation} \label{eq: conditions for U}
\mathbf{F} (\boldsymbol\vartheta ){\bf{U}} = {\bf{0}}_2^\mathsf{T} ,\,\, {{\bf{U}}^\mathsf{T}}{\bf{U}} = {{\mathbf{I}}_2}.
\end{equation} 
\end{definition}

Using these preliminaries, we can derive the covariance matrix of the proposed estimator:

\begin{proposition}
\label{prop:CCRB}
Under the small-noise conditions C1--C2, the proposed constrained weighted least-squares estimator is asymptotically unbiased to first order, and its covariance approaches:
\begin{equation} \label{eq : CCRB stoica formula}
\boldsymbol\Pi (\boldsymbol\vartheta) = {\bf{U}}{\left( {{{\bf{U}}^\mathsf{T}}{\bf{JU}}} \right)^{ - 1}}{{\bf{U}}^\mathsf{T}}
\end{equation}
where $\mathbf{U}$ is defined in \eqref{eq: conditions for U} and $\mathbf{J}$ is the FIM associated with the linear observation model \eqref{eq:d-A noisy model} which can be easily obtained as \cite{kay}
\begin{align} \label{eq: formula for J}
{\bf{J}} & \approx \mathbf{A}^\mathsf{T}(\mathbf{BC_mB}^\mathsf{T})^{-1}\mathbf{A}
\end{align}
where approximation is valid under small noise conditions:
C1) $|\Delta\tau_k|\ll\tau_k, \forall k$ and C2) $|\Delta\theta|\approx 0$.

Note that $\boldsymbol\Pi (\boldsymbol\vartheta)$ is known as the constrained CRB for the any constrained estimators of the form \eqref{eq:CWLS GTRS} \cite{Stoica_CCRB_1998}.
\end{proposition}
 
\begin{IEEEproof}
Under C1--C2, the pseudo-linear model in \eqref{eq:d-A noisy model} is accurate to first order and has covariance $\mathbf B\mathbf C_m\mathbf B^\mathsf T$. Applying the standard constrained-CRB expression \cite{Stoica_CCRB_1998,Moore_2008} to this local Gaussian model yields \eqref{eq : CCRB stoica formula}.
\end{IEEEproof}
 
Based on Proposition~\ref{prop:CCRB}, the covariance matrix of $\mathbf{\hat p}$ is simply obtained as
\begin{equation} \label{eq : matrix C definition}
\text{cov}({\hat{\bf{p}}}) = \boldsymbol\Pi(\boldsymbol{\vartheta})(1\!:\!2,1\!:\!2).
\end{equation}

Now, we aim to prove that the proposed estimator can achieve the target position CRB.

\begin{theorem}
Assume C1--C2 and that the first-stage ToA/AoA estimator is asymptotically efficient so that $\mathbf C_m$ approaches \eqref{eq:C_m}. Then, to first order in the measurement errors, the covariance matrix of the proposed estimator in \eqref{eq : matrix C definition} approaches the position CRB in \eqref{eq:CRB(p)}.
\end{theorem}

\begin{IEEEproof} 
Let us first rewrite the covariance matrix in a more compact form.
As stated in \cite{Moore_2008}, we may replace ${\bf{U}}$ with any ${\bf{V}}$ in \eqref{eq : CCRB stoica formula} where $\mathbf{F} (\boldsymbol\vartheta ){\bf{U}} = {\bf{0}}_2^\mathsf{T}$ and $\text{rank}({\mathbf{V}}) = \text{rank}({\mathbf{U}}) = 2$. One viable choice for $\bf{V}$ that satisfies both of these conditions is
\begin{equation} \label{eq : Viable chioce for V}
{\bf{V}} = \big[{{\bf{I}}_2},2\mathbf{p}\big]^\mathsf{T}
\end{equation}
Consequently, inserting the above value of $\mathbf{V}$ in \eqref{eq : CCRB stoica formula} instead of $\mathbf{U}$ yields
\begin{equation} \label{eq : C}
\text{cov}({\hat{\bf{p}}}) = {\left( {{{\bf{V}}^\mathsf{T}}{\bf{JV}}} \right)^{ - 1}}
\end{equation}

Now, we shall show that the derived covariance matrix and the main CRB associated with the observation model \eqref{eq:Received Signal Radar} are in fact equal.

By substituting \eqref{eq: formula for J} into \eqref{eq : C} and also substituting $\mathbf{C_m}$ from \eqref{eq:C_m}, we have
\begin{equation} \label{eq: Alternate form of C}
\text{cov}({\hat{\bf{p}}}) \approx {\left(\mathbf{Z}^\mathsf{T}\mathbf{\Psi}_1\mathbf{Z} - \mathbf{Z}^\mathsf{T}\mathbf{\Psi}_2\mathbf{\Psi}_3^{-1}\mathbf{\Psi}_2^\mathsf{T}\mathbf{Z}\right)^{ - 1}}
\end{equation}
under condition C1-C2, where $\mathbf{Z}=\mathbf{B}^{-1}\mathbf{AV}$.

It is evident that \eqref{eq: Alternate form of C} has a similar form as \eqref{eq:CRB(p)}. After straightforward matrix multiplications it can be shown that $\mathbf{Z} =  \mathbf{\Xi}^\mathsf{T}$,
which immediately completes the proof. 
\end{IEEEproof}

\section{Numerical Simulations}\label{sec:Numerical Simulations}

We evaluate the proposed framework using the simulation parameters in
Table~\ref{TabSettSim}. Unless otherwise specified, these parameters are
used throughout this section.

\begin{table}[t]
    \caption{General simulation settings}
    \label{TabSettSim}
    \centering
    \footnotesize
    \begin{tabular}{p{0.53\columnwidth}c}
        \hline
        \textbf{Setting} & \textbf{Value} \\ \hline
        Carrier frequency / wavelength
        & $30$ GHz / $1$ cm \\
        Bandwidth / subcarrier spacing
        & $200$ MHz / $195.1$ kHz \\
        OFDM index / frame length
        & $N=1024$ / $M=30$ \\
        TX / RX antennas
        & $N_t=24$ / $N_r=20$ \\
        RIS elements / spacing
        & $N_s=100$ / $d=\lambda/2$ \\
        Power constraint / noise PSD
        & $P_\mathrm{T}=1$ W / $-170$ dBm/Hz \\
        Target RCS / path loss factor
        & $1~\mathrm{m}^2$ / $4$ dB \\
        User SINR threshold
        & $\varepsilon=15$ dB \\ \hline
    \end{tabular}
\end{table}

\subsection{Optimal Beampattern}

We first examine the beamforming solution for a target located $150$ m
from the BS at an azimuth of $45^\circ$ and two RIS-enabled users located
$100$ m from the BS at $20^\circ$ and $60^\circ$. With optimized RIS
phases, Fig.~\ref{FigPattern1A} shows prominent beams toward both RISs
and a smaller beam toward the target. Since the direct path is considerably stronger than the indirect BS--RIS--target paths, the A-optimal design in \eqref{eq:Opt Problem General Form} allocates more transmit gain toward the RISs to strengthen these otherwise weak but geometrically useful measurements. For this geometry, more gain is assigned to the $20^\circ$ path because its localization-information direction is more complementary to the direct path.

With random RIS phases, Fig.~\ref{FigPattern1B} shows that the optimizer
primarily illuminates the $20^\circ$ RIS and the direct path, while the
$60^\circ$ path is effectively discarded because its high path loss
makes additional power allocation inefficient. Consequently, the PEB
increases from $1.58$ m with optimized RIS phases to $14.44$ m with
random phases, illustrating the importance of passive beamforming.

\begin{figure}[t]
    \centering
    \subfloat[Optimal RIS phases]{
        \includegraphics[width=0.95\columnwidth]{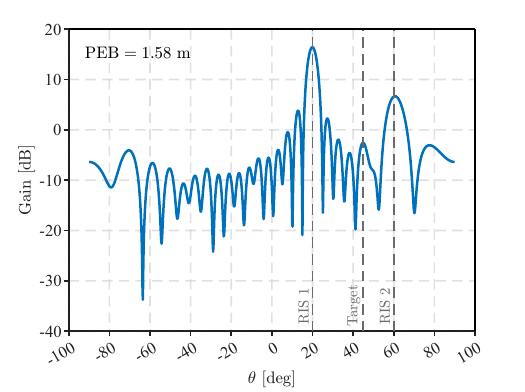}
        \label{FigPattern1A}}
    \hfill
    \subfloat[Random RIS phases]{
        \includegraphics[width=0.95\columnwidth]{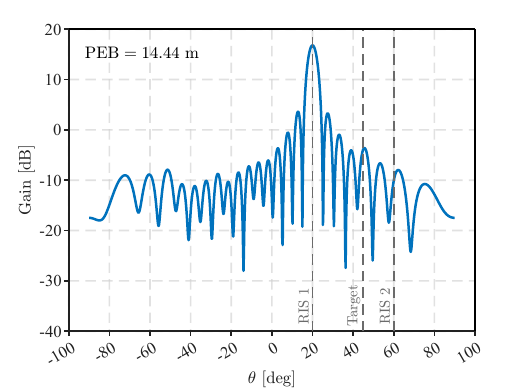}
        \label{FigPattern1B}}
    \caption{Optimized BS beampattern with (a) optimized and
    (b) random RIS phase shifts.}
    \label{FigPattern1}
\end{figure}

In another scenario, we repeated the experiment with additional RIS-enabled users, and the resulting beampattern is shown in Fig. \ref{FigPattern10}. The RIS phases are configured to be optimal.
We considered a target at $ \bp = [200, 0] $ ($ \theta=0^\circ $) and 10 RIS-enabled users positioned 100 meters from the BS at azimuth angles $ -75^\circ, -60^\circ, -45^\circ, -30^\circ, -15^\circ, 15^\circ, 30^\circ, 45^\circ, 60^\circ, 75^\circ$. The resulting beampattern, shown in Fig. \ref{FigPattern10}, reveals distinct beams directed toward the users’ angular positions. 

\begin{figure}[!htp]
	\centerline{\includegraphics[width = \columnwidth]{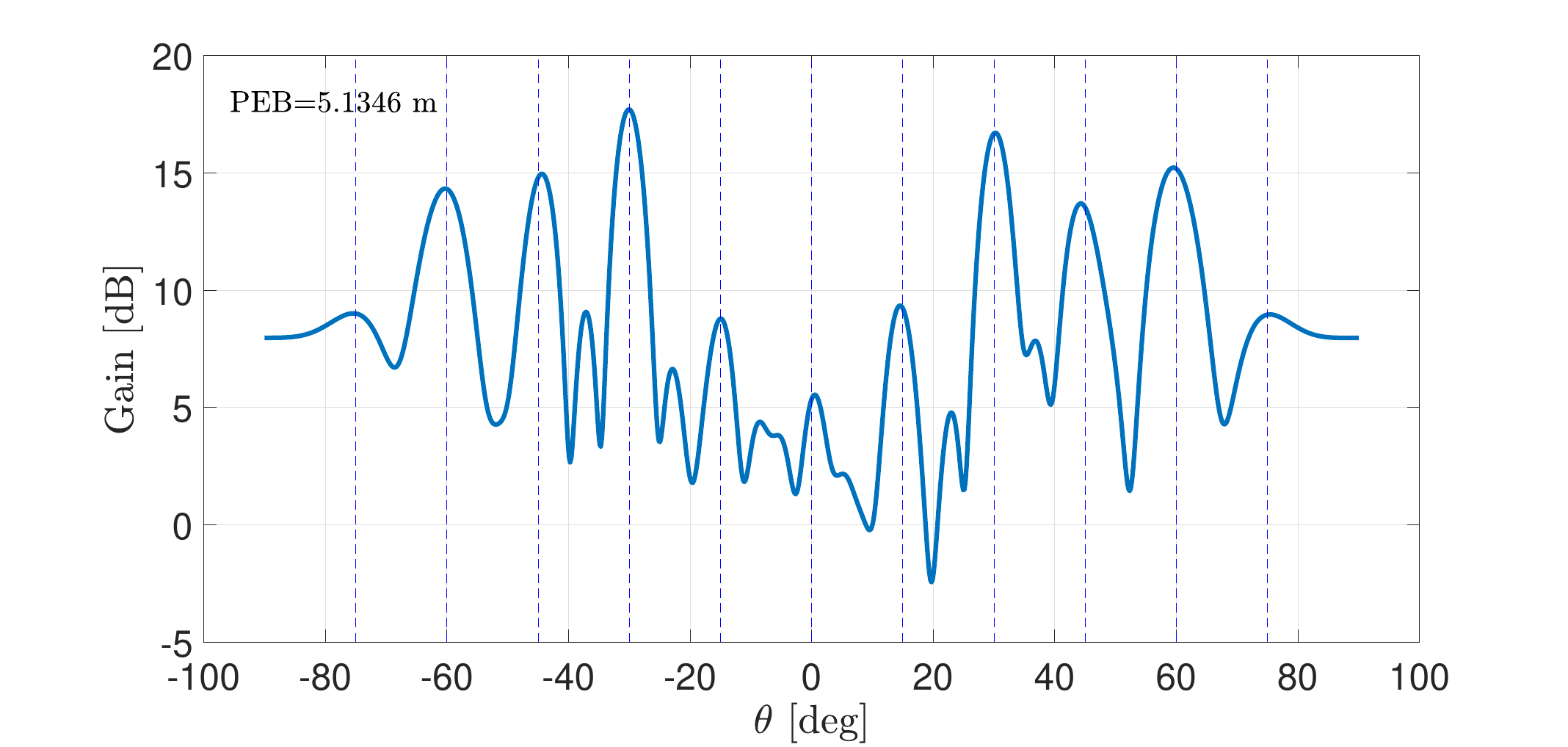}}
	\caption{The resulting beam-pattern in the presence of a target at $ \theta = 0^\circ $ and 10 RIS-enabled users at azimuth angles $ -75^\circ, -60^\circ, -45^\circ, -30^\circ, -15^\circ, 15^\circ, 30^\circ, 45^\circ, 60^\circ, 75^\circ$.}
	\label{FigPattern10}
\end{figure}

\subsection{Geometric Analysis}

We next illustrate the role of bistatic geometry. A RIS-enabled user is
placed at $[100,0]$ m, while the target position is varied over
$x\in[0,300]$ m and $y\in[-150,150]$ m. As shown in
Fig.~\ref{FigMapA}, localization accuracy deteriorates along the
half-line extending from the RIS toward positive $x$. In this region,
the circular fixed-delay contour of the direct path becomes nearly
tangent to the elliptical contour of the RIS-assisted path, so that
their ranging-information directions become nearly dependent.

Fig.~\ref{FigMapB} considers two RIS-enabled users. The additional
bistatic measurement provides another information direction and
substantially reduces the low-accuracy region. The gain therefore
depends not only on received power but also on the BS--RIS--target
geometry, suggesting that the PEB can serve as a criterion for
cooperative-RIS placement.

\begin{figure}[!t]
    \centering
    \subfloat[]{
        \includegraphics[width=0.43\textwidth]
        {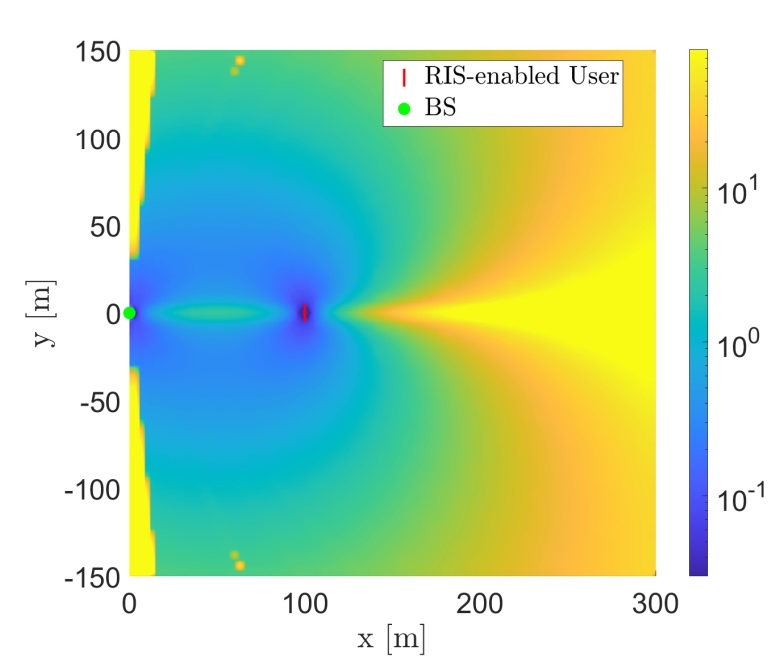}
        \label{FigMapA}}
    \hfill
    \subfloat[]{
        \includegraphics[width=0.43\textwidth]
        {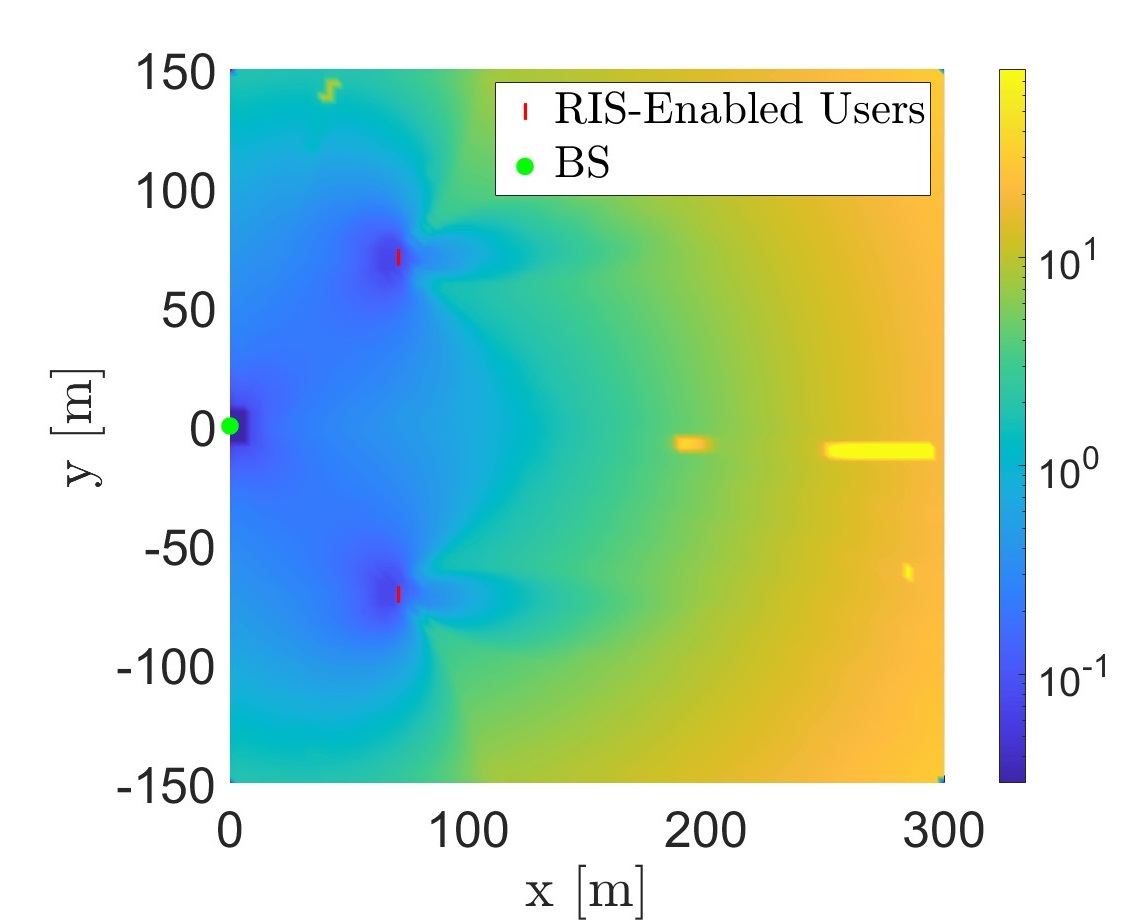}
        \label{FigMapB}}
    \caption{PEB obtained with the proposed beamforming method:
    (a) one RIS-enabled user and (b) two RIS-enabled users.}
    \label{FigMap}
\end{figure}

\subsection{ISAC Trade-off Analysis}

We investigate the effect of cooperative and ordinary users for a target
at $[150,150]$ m. User locations are independently generated within a
$200\times200$ m region and the results are averaged over Monte Carlo
realizations. To remain consistent with the ranging-dominant
beamforming model of
Section~\ref{sec:Dual-Functional Joint Active and Passive Beamforming}
and avoid a degenerate range-only geometry, the first user is
RIS-enabled in both cases. Here, the SINR threshold is
$\varepsilon=30$ dB.

As shown in Fig.~\ref{FigRmsVsK}, adding ordinary users degrades
localization because additional communication constraints consume
transmit resources without providing sensing information. In contrast,
adding RIS-enabled users improves localization since these users both
receive communication service and create additional bistatic sensing
paths. Since one RIS-enabled user is present in both cases, this is a
partially non-cooperative comparison rather than a pure BS-only
benchmark.

Fig.~\ref{FigRmsVsEpslion} further evaluates the effect of the SINR
requirement using eight users, where the first user is RIS-enabled and
the remaining users are either all ordinary or RIS-enabled. Increasing
$\varepsilon$ causes a considerably larger PEB degradation in the
ordinary-user case. With cooperative RIS-enabled users, the sensing and
communication objectives are more closely aligned, making localization
less sensitive to increasingly stringent communication requirements.


\begin{figure}[!t]
    \centering
    \includegraphics[width=0.95\columnwidth]{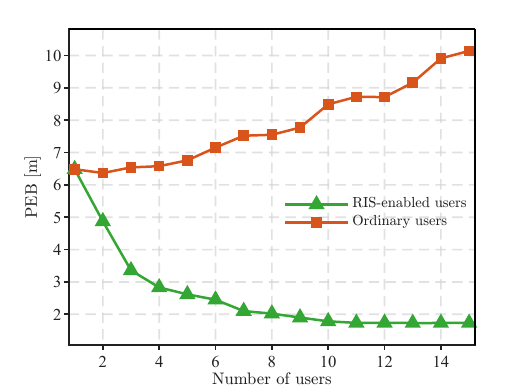}
    \caption{PEB versus the number of users for cooperative RIS-enabled and ordinary-user configurations.}
    \label{FigRmsVsK}
\end{figure}

\begin{figure}[!t]
    \centering
    \includegraphics[width=0.95\columnwidth]{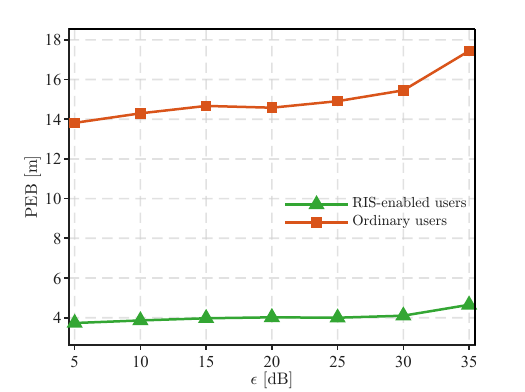}
    \caption{PEB versus the minimum communication SINR requirement.}
    \label{FigRmsVsEpslion}
\end{figure}

\subsection{Impact of System Parameters}

Finally, Fig.~\ref{fig:parametric} examines the parameters most directly
connected to the analytical results. Figs.~\ref{fig:peb_BW} and
\ref{fig:peb_pt} show that localization accuracy improves with bandwidth
and transmit power. In the ranging-dominant regime, the analysis
predicts approximately $1/B$ and $1/\sqrt{P_\mathrm{T}}$ PEB trends,
respectively. The RMSE of the proposed estimator also approaches the CRB
as the measurement quality improves, corroborating the first-order
efficiency result of Section~\ref{sec:Efficient Estimator Design}.

Fig.~\ref{fig:peb_Nris} directly verifies the RIS scaling laws derived
in Section~\ref{sec:Estimation-Theoretic Fundamental Limits}. On the
logarithmic axes, the PEB decreases approximately as $1/N_s$ with
coherently optimized RIS phases and as $1/\sqrt{N_s}$ with independent
random phases, in agreement with the predicted $N_s^2$ and $N_s$
Fisher-information gains, respectively.

\begin{figure*}[!t]
    \centering
    \subfloat[Signal bandwidth]{
        \includegraphics[width=0.30\textwidth]
        {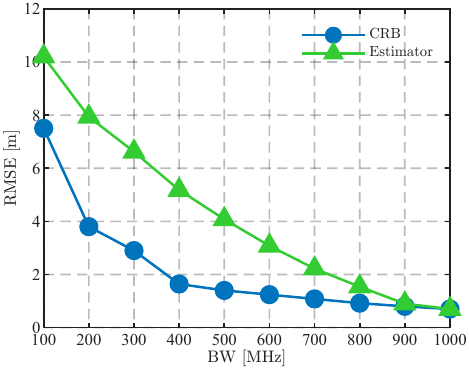}
        \label{fig:peb_BW}}
    \hfill
    \subfloat[Transmit power]{
        \includegraphics[width=0.30\textwidth]
        {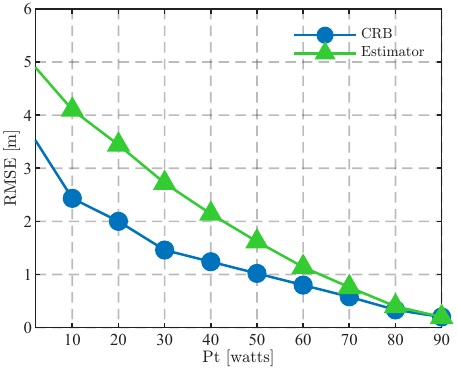}
        \label{fig:peb_pt}}
    \hfill
    \subfloat[Number of RIS elements]{
        \includegraphics[width=0.30\textwidth]
        {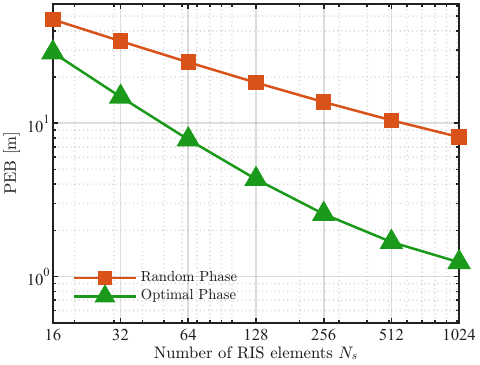}
        \label{fig:peb_Nris}}
    \caption{Localization performance versus key system parameters:
    (a) signal bandwidth, (b) transmit power, and
    (c) number of RIS elements.}
    \label{fig:parametric}
\end{figure*}

\section{Conclusion}\label{sec:Conclusion}
In this paper, we proposed a RIS-aided ISAC framework for joint target localization and multi-user communication in a 2D scenario. Theoretical analysis revealed that random RIS phase assignments improve Fisher information by a factor of $N_s$, while optimized phases achieve a gain of $N_s^2$. We designed an estimator combining AoA, direct ToA, and $K$ RIS-assisted indirect ToA measurements, whose covariance approaches the CRB to first order under small measurement errors and efficient first-stage ToA/AoA estimation. To balance sensing-communication trade-offs, we formulated a non-convex optimization problem for joint RIS phase and beamforming design and handled it using an iterative SCA procedure with convex subproblems. Simulations validated the theoretical bounds.

\appendices
\section{Calculation of the FIM} \label{sec:Calculation of the FIM}
Based on complex Gaussian distribution of the received signal in \eqref{eq:Received Signal Radar}, the elements of $\mathbf{J}_{\boldsymbol{\eta}}$, which is a $(3K+4)\times (3K+4)$ matrix, can be found as \cite{kay}
\begin{align}\label{eq:FIM Elements Def.}
\mathbf{J}_{\boldsymbol{\eta}}(x,y) &=
\mathbb{E}_{\mathbf{y}|\boldsymbol{\eta}}\left\{-\frac{\partial^2 \ln f(\mathbf{y}|\boldsymbol{\eta})}{\partial x \partial y}\right\} \nonumber\\
&= \frac{1}{\sigma^2_R} \sum_{n=-\frac{N}{2}}^{\frac{N}{2}}\Re\left\{\frac{\partial \boldsymbol{\mu}^\mathsf{H}[n]}{\partial x} \frac{\partial \boldsymbol{\mu}[n]}{\partial y}\right\}
\end{align}
where $f(\mathbf{y}|\boldsymbol{\eta})$ is the corresponding conditioned likelihood function\footnote{Note that $\mathbf{y}$ is the vectorized version of $\mathbf{Y}_{\mathrm{R}}$ for all instances and subcarriers.}, and $\boldsymbol{\mu}[n] = \sum_{k=0}^{K} \beta_k e^{-j2\pi n\Delta f \tau_k} \text{vec}\left\{\mathbf{G}_k \mathbf{X}[n]\right\}$.

Therefore, the elements of $\mathbf{\Psi}_1$ in \eqref{eq:J_eta partiotioned} can be expressed as
\begin{align} \label{eq:Psi_1}
\mathbf{\Psi}_1 & = 
\left[ \begin{matrix}
\mathbf{\Psi}_{1,1} & \boldsymbol{\psi}_{1,2}  \\
\boldsymbol{\psi}_{1,2}^{T} & \psi_{1,3}  \\
\end{matrix} \right]
\end{align}
\vspace{-.2cm}
where
\begin{align}
&\mathbf{\Psi}_{1,1}(i\!+\!1,k\!+\!1)  = \Re\left\{\beta_i^{\ast}\beta_k \mathrm{Tr}(\mathbf{G}_k\mathbf{R_X}\mathbf{G}_i^\mathsf{H}) f_2(\tau_i\!-\!\tau_k)\right\} \nonumber\\
&\boldsymbol{\psi}_{1,2}(i\!+\!1)  = \sum_{k=0}^{K} \Im \left\{\beta_i^{\ast}\beta_k \mathrm{Tr}(\mathbf{\dot{G}}_k\mathbf{R_X}\mathbf{G}_i^\mathsf{H}) f_1(\tau_i\!-\!\tau_k)\right\} \nonumber\\
&\psi_{1,3} = \sum_{i=0}^{K}\sum_{k=0}^{K} \Re \left\{\beta_i^{\ast}\beta_k \mathrm{Tr}(\mathbf{\dot{G}}_k\mathbf{R_X}\mathbf{\dot{G}}_i^\mathsf{H}) f_0(\tau_i\!-\!\tau_k)\right\} \nonumber \\
&i,k=0,1,\ldots,K
\end{align}
where $\mathbf{R_X}=\mathbb{E}\{\frac{1}{M}\mathbf{X}[n]\mathbf{X}^\mathsf{H}[n]\}=\mathbf{WW}^\mathsf{H}$ under the second-order stream model in Sec.~\ref{subsec:Transmitter Model}. Treating the scalar path gains as locally fixed nuisance amplitudes, the derivatives of $\mathbf{G}_k$ with respect to $\theta$ are
\begin{align}\label{eq: G_k dot}
\mathbf{\dot{G}}_0 & = \gamma_0[\mathbf{\dot{c}}(\theta)\mathbf{a}^\mathsf{H}(\theta) + \mathbf{c}(\theta)\mathbf{\dot{a}}^\mathsf{H}(\theta)] \nonumber\\
\mathbf{\dot{G}}_k & = \gamma_k \mathbf{\dot{c}}(\theta)\mathbf{a}^\mathsf{H}(\varphi_k), \quad k=1,2,\ldots,K 
\end{align}

Furthermore, the elements of $\mathbf{\Psi}_2$ in \eqref{eq:J_eta partiotioned} can be written as
\begin{align}\label{eq:Psi_2}
\mathbf{\Psi}_2 & = \left[ \mathbf{\Psi}_{2,1}^\mathsf{T}, \boldsymbol{\psi}_{2,2} \right]^\mathsf{T}
\end{align}
\vspace{-.2cm}
where
\begin{align}
&\mathbf{\Psi}_{2,1}(i\!+\!1,m)  = \Im\left\{\beta_i^{\ast} \mathrm{Tr}(\mathbf{G}_k\mathbf{R_X}\mathbf{G}_i^\mathsf{H}) f_1(\tau_i\!-\!\tau_k) \mathbf{v}^\mathsf{T}\right\} \nonumber\\
&\boldsymbol{\psi}_{2,2}(m)  = \sum_{i=0}^{K} \Re \left\{\beta_i^{\ast} \mathrm{Tr}(\mathbf{G}_k\mathbf{R_X}\mathbf{\dot{G}}_i^\mathsf{H}) f_0(\tau_i\!-\!\tau_k)\mathbf{v}\right\} \nonumber\\
&m=2k\!+\!1\!:\!2k\!+\!2,\quad i,k=0,1,\ldots,K
\end{align}
where $\mathbf{v}=[1,j]^\mathsf{T}$.

Finally, the elements of $\mathbf{\Psi}_3$ in \eqref{eq:J_eta partiotioned} is given by
\begin{align}
&\mathbf{\Psi}_{3}(l,m)  = \Re\left\{\mathrm{Tr}(\mathbf{G}_k\mathbf{R_X}\mathbf{G}_i^\mathsf{H}) f_0(\tau_i\!-\!\tau_k) \mathbf{Q}\right\} \\
&l=2i\!+\!1\!:\!2i\!+\!2,\quad m=2k\!+\!1\!:\!2k\!+\!2,\quad i,k=0,1,\ldots,K
\nonumber
\end{align}
where $\mathbf{Q}=\left[ \begin{matrix}
1 & j  \\
-j & 1  \\
\end{matrix} \right]$.
In the above equations, the functions $f_i(z), i=0,1,2$ are defined as
\begin{align}\label{eq:f_i(z)}
f_i(z) \triangleq \frac{M}{\sigma_R^2}\sum_{n=-\frac{N}{2}}^{\frac{N}{2}} (2\pi n\Delta f)^i e^{j2\pi n\Delta f z}
\end{align}
\vspace{-.7cm}
\section{Proof of Theorem 1} \label{sec:Proof of Theorem 1}
By direct matrix multiplication, the first term in the EFIM in \eqref{eq:CRB(p)} becomes
\begin{align} \label{eq: First term EFIM}
\mathbf{\Xi}&\mathbf{\Psi}_1\mathbf{\Xi}^\mathsf{T}\nonumber\\
& = \frac{1}{c^2}\sum_{i=0}^{K} \sum_{k=0}^{K}
\Re\left\{\beta_i^{\ast}\beta_k \mathrm{Tr}(\mathbf{G}_k\mathbf{R_X}\mathbf{G}_i^\mathsf{H}) f_2(\tau_i\!-\!\tau_k)\right\} \mathbf{e}_i\mathbf{e}_k^\mathsf{T} \nonumber\\
& + \frac{1}{r^2}\sum_{i=0}^{K}\sum_{k=0}^{K} \Re \left\{\beta_i^{\ast}\beta_k \mathrm{Tr}(\mathbf{\dot{G}}_k\mathbf{R_X}\mathbf{\dot{G}}_i^\mathsf{H}) f_0(\tau_i\!-\!\tau_k)\right\}\boldsymbol{\rho}\boldsymbol{\rho}^\mathsf{T} \nonumber\\
& + \frac{1}{rc}\sum_{i=0}^{K}\sum_{k=0}^{K} \Im \left\{\beta_i^{\ast}\beta_k \mathrm{Tr}(\mathbf{\dot{G}}_k\mathbf{R_X}\mathbf{G}_i^\mathsf{H}) f_1(\tau_i\!-\!\tau_k)\right\}\nonumber\\
&\hspace{35mm}\times\left(\mathbf{e}_i\boldsymbol{\rho}^\mathsf{T}
+\boldsymbol{\rho}\mathbf{e}_i^\mathsf{T}\right).
\end{align}

When paths are resolvable, i.e., $|\tau_i-\tau_k|>\frac{1}{B}$, the overlapping terms become negligible so that the terms related to $i=k$ are dominant. Therefore, \eqref{eq: First term EFIM} can be approximately written as
\vspace{-.5cm}
\begin{align} \label{eq: First term EFIM Approximate}
\mathbf{\Xi}\mathbf{\Psi}_1\mathbf{\Xi}^\mathsf{T}  &\approx
\frac{1}{c^2}\sum_{k=0}^{K} |\beta_k|^2\mathrm{Tr}(\mathbf{G}_k\mathbf{R_X}\mathbf{G}_k^\mathsf{H}) f_2(0) \mathbf{e}_k\mathbf{e}_k^\mathsf{T} \nonumber\\
& + \frac{1}{r^2} \sum_{k=0}^{K}|\beta_k|^2\mathrm{Tr}(\mathbf{\dot{G}}_k\mathbf{R_X}\mathbf{\dot{G}}_k^\mathsf{H}) f_0(0) \boldsymbol{\rho}\boldsymbol{\rho}^\mathsf{T} \nonumber\\
&+ \frac{1}{rc} \sum_{k=0}^{K}|\beta_k|^2\Im\{\mathrm{Tr}(\mathbf{\dot{G}}_k\mathbf{R_X}\mathbf{G}_k^\mathsf{H})\} f_1(0)\nonumber\\
&\hspace{35mm}\times\left(\mathbf{e}_k\boldsymbol{\rho}^\mathsf{T}+
\boldsymbol{\rho}\mathbf{e}_k^\mathsf{T}\right)
\end{align}

The second term of EFIM in \eqref{eq:CRB(p)} has two ingredients, i.e., $\mathbf{\Xi}\mathbf{\Psi}_2$ and $\mathbf{\Psi}_3^{-1}$, which can be approximated, in a similar manner, as
\begin{align}
&\mathbf{\Xi}\mathbf{\Psi}_2(:,m)  \approx \frac{1}{c}\mathrm{Tr}(\mathbf{G}_k\mathbf{R_X}\mathbf{G}_k^\mathsf{H}) f_1(0) \mathbf{e}_k [-\beta_k^{\mathcal{I}},\beta_k^{\mathcal{R}}]\nonumber\\
& \qquad\qquad\qquad+\frac{1}{r}\Re\left\{\beta_k^{\ast} \mathrm{Tr}(\mathbf{G}_k\mathbf{R_X}\mathbf{\dot{G}}_k^\mathsf{H}) f_0(0) \boldsymbol{\rho}\mathbf{v}^\mathsf{T}\right\}\\
&\mathbf{\Psi}_3^{-1} \approx \mathbf{\Lambda} \otimes \mathbf{I}_2, \,\, \mathbf{\Lambda}(k\!+\!1,k\!+\!1)=(\mathrm{Tr}(\mathbf{G}_k\mathbf{R_X}\mathbf{G}_k^\mathsf{H}) f_0(0))^{-1}
\end{align}
where $\mathbf{\Lambda}$ is a diagonal matrix and $m=2k\!+\!1\!:\!2k\!+\!2,\, k=0,1,\ldots,K$.

After some straightforward mathematical manipulations, the corresponding EFIM in \eqref{eq:CRB(p)} becomes
\begin{align}
\mathbf{J}_\text{e}(\mathbf{p}) &= \frac{1}{c^2}\sum_{k=0}^{K} |\beta_k|^2 [f_2(0)\!-\!f_1^2(0) f_0^{-1}(0)] \ell_k \mathbf{e}_k\mathbf{e}_k^\mathsf{T} \nonumber\\
& + \frac{1}{r^2} \sum_{k=0}^{K}|\beta_k|^2 f_0(0)(m_k-\frac{|n_k|^2}{\ell_k}) \boldsymbol{\rho}\boldsymbol{\rho}^\mathsf{T}
\end{align}
where $\ell_k$, $m_k$ and $n_k$ are defined in Theorem 1.
Denoting $\mathbf{J}_\text{r}(\theta,\psi_k,s_k)=\mathbf{e}_k\mathbf{e}_k^\mathsf{T}$ and $\mathbf{J}_{\text{a}}(\theta)=\boldsymbol{\rho}\boldsymbol{\rho}^\mathsf{T}$, and computing $f_i(0)$'s from \eqref{eq:f_i(z)} completes the proof.

Now, let us simplify the variables $\ell_k, m_k$ and $n_k$ in \eqref{eq:a_i def}-\eqref{eq:c_i def}.
Let us denote $\mathbf{a}(\theta)$, $\mathbf{c}(\theta)$ and $\mathbf{a}(\varphi_k)$ by $\mathbf{a}_{\theta}$, $\mathbf{c}_\theta$ and $\mathbf{a}_{\varphi_k}$, respectively, and so does for their derivatives.
By denoting $\mathbf{R_X}=\sum_{i=1}^{\tilde{K}}\mathbf{w}_i\mathbf{w}_i^\mathsf{H}$, and using definition of $\mathbf{G}_k$ and $\mathbf{\dot{G}}_k$ in \eqref{eq: Channel G_k} and \eqref{eq: G_k dot}, we have
\begin{align}\label{eq:Simplified l_k,m_k,n_k}
\!\!\ell_0 & = \gamma_0^2\mathrm{Tr}(\mathbf{c}_\theta\mathbf{a}^\mathsf{H}_\theta\sum_{i=1}^{\tilde{K}}\mathbf{w}_i\mathbf{w}_i^\mathsf{H}\mathbf{a}_\theta\mathbf{c}^\mathsf{H}_\theta) = \gamma_0^2\|\mathbf{c}_\theta\|^2\sum_{i=1}^{\tilde{K}}|\mathbf{a}^\mathsf{H}_\theta\mathbf{w}_i|^2 \nonumber\\
\!\!\ell_k & = \!|\gamma_k^{}|^2 \mathrm{Tr}(\mathbf{c}_\theta\mathbf{a}^\mathsf{H}_{\varphi_k}\sum_{i=1}^{\tilde{K}}\!\mathbf{w}_i\mathbf{w}_i^\mathsf{H}\mathbf{a}_{\varphi_k}\mathbf{c}^\mathsf{H}_\theta) \!=\! |\gamma_k^{}|^2 \|\mathbf{c}_\theta\|^2\sum_{i=1}^{\tilde{K}}|\mathbf{a}^\mathsf{H}_{\varphi_k}\mathbf{w}_i|^2 \nonumber\\
\!\!m_0 & = \gamma_0^2\mathrm{Tr}((\mathbf{\dot{c}}_\theta\mathbf{a}^\mathsf{H}_\theta+\mathbf{c}_\theta\mathbf{\dot{a}}^\mathsf{H}_\theta)\sum\nolimits_{i=1}^{\tilde{K}}\!\!\mathbf{w}_i\mathbf{w}_i^\mathsf{H}(\mathbf{a}_\theta\mathbf{\dot{c}}_\theta^\mathsf{H}+\mathbf{\dot{a}}_\theta\mathbf{c}_\theta^\mathsf{H})) \nonumber\\&= \gamma_0^2[\|\mathbf{\dot{c}}_\theta\|^2\sum_{i=1}^{\tilde{K}}|\mathbf{a}^\mathsf{H}_\theta\mathbf{w}_i|^2 + \|\mathbf{c}_\theta\|^2\sum_{i=1}^{\tilde{K}}|\mathbf{\dot{a}}^\mathsf{H}_\theta\mathbf{w}_i|^2] \nonumber\\
\!\!m_k & = \!|\gamma_k^{}|^2 \mathrm{Tr}(\mathbf{\dot{c}}_\theta\mathbf{a}^\mathsf{H}_{\varphi_k}\!\sum_{i=1}^{\tilde{K}}\!\mathbf{w}_i\mathbf{w}_i^\mathsf{H}\mathbf{a}_{\varphi_k}\mathbf{\dot{c}}^\mathsf{H}_\theta) \!=\! |\gamma_k^{}|^2 \|\mathbf{\dot{c}}_\theta\|^2\sum_{i=1}^{\tilde{K}}|\mathbf{a}^\mathsf{H}_{\varphi_k}\mathbf{w}_i|^2 \nonumber\\
\!\!n_0 & = \gamma_0^2\mathrm{Tr}(\mathbf{c}_\theta\mathbf{a}^\mathsf{H}_\theta\sum\nolimits_{i=1}^{\tilde{K}}\!\!\mathbf{w}_i\mathbf{w}_i^\mathsf{H}(\mathbf{a}_\theta^\mathsf{H}\mathbf{\dot{c}}_\theta+\mathbf{\dot{a}}_\theta^\mathsf{H}\mathbf{c}_\theta)) \nonumber\\&= \gamma_0^2\|\mathbf{c}_\theta\|^2\sum\nolimits_{i=1}^{\tilde{K}}\!\!\mathbf{a}^\mathsf{H}_\theta\mathbf{w}_i\mathbf{w}_i^\mathsf{H}\mathbf{\dot{a}}_\theta \nonumber\\
\!\!n_k & = |\gamma_k^{}|^2 \mathrm{Tr}(\mathbf{c}_\theta\mathbf{a}^\mathsf{H}_{\varphi_k}\sum_{i=1}^{\tilde{K}}\mathbf{w}_i\mathbf{w}_i^\mathsf{H}\mathbf{a}_{\varphi_k}\mathbf{\dot{c}}^\mathsf{H}_\theta) = 0
\end{align}
for $k = 1,2,\ldots,K$, where, for the centered ULA, the steering vector is orthogonal to its own derivative (e.g., $\mathbf{a}_\theta^\mathsf{H}\mathbf{\dot{a}}_\theta=0$ and $\mathbf{c}_\theta^\mathsf{H}\mathbf{\dot{c}}_\theta=0$). Steering vectors at different angles are not assumed orthogonal.

\bibliographystyle{IEEEtran}
\bibliography{refs}
\end{document}